\documentclass{LMCS}

\def\doi{8(3:20)2012}
\lmcsheading%
{\doi}
{1--31}
{}
{}
{Oct.~15, 2011}
{Sep.~20, 2012}
{}
 
\usepackage{enumerate}
\usepackage{hyperref}

\usepackage{amsmath,amssymb,amsfonts,latexsym,stmaryrd}
\usepackage{epsfig}

\usepackage{gastex,delarray,wrapfig}
\usepackage{extarrows}

\usepackage{verbatim}
\usepackage[mathscr]{euscript}
\usepackage{subfigure,multimedia,mathtools}
\usepackage{delarray}
\usepackage{graphicx}
\usepackage{tikz}
\usepackage{listings}
\usepackage{float}

\usepackage{pgf}
\usepackage{tikz}
\usetikzlibrary{automata}
\usetikzlibrary{trees}
\usetikzlibrary{shapes}
\usetikzlibrary{petri}
\usetikzlibrary{arrows}
\usetikzlibrary{snakes}
\usetikzlibrary{backgrounds}

\def\by#1{\mathop{{\hbox{\setbox0=\hbox{$\scriptstyle{#1\quad}$}{$%
\mathrel{\mathop{\setbox1=\hbox to \wd0{\rightarrowfill}\ht1=3pt\dp1=-2pt\box1}\limits^{#1}}%
$}}}}}

\newcommand{\lby}[1]{{\xLongrightarrow  
{#1}}}
\newcommand{\bby}[2]{\xLongrightarrow  
[#1]{#2}}

\newcommand{\m}[1]{\mathcal{#1}}
\newcommand{\co}{\langle}
\newcommand{\cf}{\rangle}

\newcommand{\ra}{\rightarrow}
\newcommand{\Ra}{\Rightarrow}

\begin{document}

\title[Model-Checking of Ordered Multi-Pushdown Automata]{ Model-Checking of Ordered Multi-Pushdown Automata\rsuper*}

\author[M.~F.~Atig]{Mohamed Faouzi Atig}	
\address{Uppsala University, Sweden}	
\email{mohamed\_faouzi.atig@it.uu.se}  



\titlecomment{{\lsuper*}A shorter version of this paper has been published in the Proceedings of CONCUR '10  and FSTTCS'10}


\begin{abstract}
   We address  the  verification problem of  ordered multi-pushdown automata: A multi-stack extension of  pushdown automata that comes with a constraint on stack transitions such that  a pop   can only be performed on the first non-empty stack. First, we show that the emptiness problem for  ordered multi-pushdown automata is in 2ETIME. Then, we prove 
  that, for an ordered multi-pushdown automata,   the set of all  predecessors of a regular set of configurations is an effectively constructible regular set. We exploit this result to solve the global model-checking  which consists in computing the set of all configurations of an ordered multi-pushdown automaton that satisfy a given $w$-regular property (expressible in linear-time temporal logics  or the linear-time $\mu$-calculus). As an immediate consequence, we obtain an 2ETIME upper bound for the model-checking problem of $w$-regular properties for ordered multi-pushdown automata (matching its
lower-bound). 

\end{abstract}

\maketitle

\section*{Introduction}
Automated verification of multi-threaded programs is an important and a highly challenging problem. In fact, even  when such programs manipulate data ranging over finite domains, their control structure can be complex due to the handling of (recursive) procedure calls in the presence of concurrency and synchronization between threads.

In the last few years, a lot of effort has been devoted to the verification problem for models of concurrent programs (see, e.g., \cite{BMOT05,MadhuLICS07,kahlon09,ABT08,TACAS08,AT09,HLMS10,LR08,GantyMM10,conf/popl/EmmiQR11,conf/sas/BouajjaniEP11,LaTorreN11}) where each thread corresponds to a sequential program with (recursive) procedure calls.  In fact, it is well admitted that 
pushdown automata   are  adequate models for such kind  of threads  \cite{EK99,RSJ03}, and  therefore, it is natural  to model  recursive concurrent programs  as multi-stack automata.

In general,  multi-stack  automata are Turing powerful and
hence come along with undecidability of basic decision problems \cite{Ram00}.  A lot of efforts have been nevertheless devoted recently to the development of  precise analysis algorithms of specific formal models of some classes of programs \cite{LS98,EP00,BT03,SV06,JhalaMajumdar07}.

 Context-bounding  has been proposed in \cite{QR05} as a suitable technique for the analysis of multi-stack automata.  The idea is to consider only runs of the automaton  that  can be divided into a given number of contexts, where in  each context pop and push transitions are exclusive to one stack.  
The state space which may  be explored is  still unbounded in presence of recursive procedure calls, but the context-bounded reachability problem is NP-complete even in this case.
In fact, context-bounding provides a very useful tradeoff between computational complexity and verification coverage.

 In
\cite{MadhuLICS07}, La Torre et al.\  propose a more general definition of the notion of a context.  For that,  they define the  class   of \emph{bounded-phase visibly
  multi-stack pushdown automata} (BVMPA) where  only those runs are taken into
consideration that can be split into a given number of phases, where each
phase admits pop transitions of one particular stack only. In the above case, the emptiness problem    is  decidable in  double exponential  time by reducing it to the emptiness problem for tree automata.

Another way to regain decidability is to  impose some order  on stack transitions. In \cite{multi96}, Breveglieri et
al.\ define \emph{ordered multi-pushdown automata} (OMPA), which impose a linear
ordering on stacks. Stack transitions are constrained in such a way
that a pop transition is reserved to the first non-empty stack.  In \cite{ABH-dlt08}, the emptiness  problem for OMPA is  shown to be 2ETIME-complete. ({Recall that  2ETIME is the class of all decision problems solvable by a deterministic Turing machine in time $2^{2^{dn}}$ for some constant $d$.}) The proof  of this result lies in an  encoding of OMPA  into  some class of grammars for which the emptiness problem is decidable.  Moreover,  the class of ordered multi-pushdown automata  with $2k$ stacks is shown to be  strictly more expressive than bounded-phase visibly   multi-stack pushdown automata with $k$ phases \cite{ABH-dlt08}.

In this paper, we consider the problem of verifying ordered multi-pushdown automata with respect to a given $w$-regular property (expressible in  the linear-time temporal logics \cite{Pnu77} or the linear-time $\mu$-calculus \cite{Var88}). In particular, we are interested  in solving the global model checking for ordered multi-pushdown automata which consists in computing the set of all configurations  that satisfy a given $w$-regular property.  The basic ingredient for achieving this goal is to define  a procedure for computing the set of backward reachable configurations from a given set of configurations. Therefore, our first task is to find a finite symbolic representation of the possibly infinite state-space of an ordered multi-pushdown automaton. For that, we consider the class of  recognizable sets of configurations defined  using  finite state automata  \cite{QR05,ABT08,anil10}.

We show   that  for an ordered multi-pushdown automaton  $\m{M}$ the set of all  predecessors $\mathit{Pre}^*(C)$ of a recognizable set of configurations $C$ is an effectively constructible recognizable set. For this, we introduce the class of {\em effective generalized pushdown automata} (EGPA) where transitions on stacks are (1) pop the top symbol of the stack, and  (2) push a word in some {\em effective language}   $L$ over the stack alphabet. The language  $L$ is said to be {\em effective}  if the problem consisting in  checking whether $L$ intersects a given regular language is decidable.  Observe that  $L$  can be    any  finite union of languages defined  by  a class of  automata closed under intersection with regular languages and  for which 
 the emptiness problem is decidable  (e.g., pushdown automata, Petri nets, lossy channel machines,  etc).
  Then, we show that the automata-based saturation procedure for computing the set of predecessors in standard pushdown automata  \cite{BEM97} can be extended to prove that for  EGPA  too the set of all predecessors of a  regular  set of configurations is  a regular set and  effectively constructible.  As an immediate consequence of this result, we obtain  similar   decidability results   of   the decision problems for EGPA  like  the ones obtained for pushdown automata.

Then,  we show that, given an OMPA $\m{M}$ with $n$ stacks, it is possible  to construct an  EGPA  $\m{P}$, whose  pushed languages are defined by   OMPA with $(n-1)$ stacks, such that  the following invariant is preserved:  The state  and the stack content of $\m{P}$ are respectively the same as the state and the content of the $n^{th}$ stack of $\m{M}$ when its first $(n-1)$ stacks are empty.  Let $C$ be a recognizable set of configurations of $\m{M}$, and $\mathit{Pre}^*(C)$ the set of predecessors of $C$.  
Then, we can apply the saturation procedure  to $\m{P}$ to show that  the set of configurations  $C_n$,  consisting of   $\mathit{Pre}^*(C)$  restricted to the configurations in which   the first $(n-1)$ empty stacks are empty, is   recognizable and effectively constructible.  To compute the intermediary configurations in  $\mathit{Pre}^*(C)$ where  the first $(n-1)$ stacks   are not empty, we construct an ordered multi-pushdown automaton  $\m{M}'$ with $(n-1)$ stacks that: $(1)$ performs  the same transitions on its stacks as the ones   performed by $\m{M}$ on its first $(n-1)$ stacks, and $(2)$ simulates a push transition  of $\m{M}$ over  its  $n^{th}$ stack  by a transition of  the finite-state automaton accepting  the recognizable  set of configurations $C_n$. Now, we can apply the induction hypothesis to $\m{M}'$ and construct a finite-state automaton accepting the set of all  predecessors $\mathit{Pre}^*(C)$.

As an application of this result, we show that the set of configurations   of an ordered multi-pushdown automaton satisfying a given $w$-regular property  is recognizable and effectively constructible. Our approach also allows us to obtain  an 2ETIME upper bound for the model checking problem of $w$-regular properties for ordered multi-pushdown automata (matching its lower-bound \cite{ABH-dlt08}).

\medskip

\noindent
{\bf Related works:}  
As mentioned earlier, context-bounding has been introduced by Qadeer and Rehof in \cite{QR05} for detecting safety bugs in shared memory concurrent programs. Several extensions of context-bounding to other classes of programs and efficient procedures for context-bounded analysis have been proposed in \cite{BESS05,BFQ07,LR08,ABQ09,TMP09,LaTorreMP-PLDI09,conf/cav/TorreMP10}. Other bounding concepts allowing for larger/incomparable coverage of the explored behaviors have been proposed in \cite{MadhuLICS07,GantyMM10,conf/popl/EmmiQR11,conf/sas/BouajjaniEP11,LaTorreN11}.

In \cite{anil10}, A. Seth   shows that  the set of predecessors of a recognizable set of configurations of a bounded-phase visibly
  multi-stack pushdown automaton is recognizable and effectively constructible. In fact, 
our results generalize the obtained  result in \cite{anil10}  since any bounded-phase visibly
  multi-stack pushdown automaton with $k$ phases can be simulated by an ordered multi-pushdown automaton  with $2k$ stacks \cite{ABH-dlt08}.

In this line of work, the focus has been  on checking safety properties.    In \cite{MadP11}, P. Madhusudan and G. Parlato propose a unified and generalized technique to show the decidability of  the emptiness problem  for several restricted classes of  concurrent pushdown automata (including  ordered multi-pushdown automata). The proof is done by showing that the graphs of each such computations (seen as a multi-nested words) have a bounded tree-width.  This result implies that model checking MSO  properties (over finite-computations) for these systems is decidable for OMPA. In the conclusion of   \cite{MadP11}, the authors claim that their approach can be used to show the decidability of the model checking  of $\omega$-regular properties over infinite computations of OMPA but no proof was provided. Moreover, the authors does not address the global model-checking problem for OMPA neither establish its  complexity  as we do.

To the best of our knowledge, this is the first work that addresses the global model checking for ordered multi-pushdown automata. In this paper, we extend \cite{atig10,fsttcs/Atig10}  by adding details and missing proofs.



\section{Preliminaries}
\label{prel}

In this  section, we  introduce  some basic definitions and notations that will be  used in the rest of the paper.

\medskip

\noindent
{\bf Integers:} Let   $\mathbb{N}$ be    the set of natural numbers. For every  $i,j \in \mathbb{N}$  such that $i \leq j$,   we use  $[i,j]$ (resp. $[i,j[$) to denote    the set $\{k \in \mathbb{N}\,|\, i\leq k \leq j\}$  (resp. $\{k \in \mathbb{N}\,|\, i\leq k < j\}$).

\medskip

\paragraph{\bf Words and languages:}
Let $\Sigma$ be a finite alphabet. We denote by $\Sigma^*$ (resp. $\Sigma^+$) the set of all words (resp. non empty words) over $\Sigma$, and by $\epsilon$ the empty word. A language is a (possibly infinite) set of words.   We use $\Sigma_{\epsilon}$  and $\mathit{Lang}(\Sigma)$ to denote  respectively  the set $\Sigma \cup \{\epsilon\}$ and  the set of all  languages over $\Sigma$. Let $u$ be a word over $\Sigma$. The length of $u$ is denoted by $|u|$. For every    $j \in  [1,|u|]$, we use $u(j)$ to denote the $j^{th} $ letter  of $u$. We denote  by   $u^R$ the mirror of $u$. 

%

\medskip

\noindent
{\bf Transition systems:}
\label{trans-syste} A transition system (TS for short) is a triplet  $\m{T}=(C,\Sigma,\rightarrow)$ where: $(1)$   $C$ is a (possibly infinite)  set of configurations, 
$(2)$ $\Sigma$ is a finite set of labels (or actions) such that $C \cap \Sigma=\emptyset$, and $(3)$  $\rightarrow \subseteq C \times \Sigma_{\epsilon} \times C$ is a transition relation.  We write  $c \by{a}_{\m{T}} c'$ whenever $c$ and $c'$ are two configurations and  $a$ is an action  such that  $(c,a,c') \in \rightarrow$.

  Given two configurations $c,c' \in C$, a finite run  $\rho$ of $\m{T}$ from $c$ to $c'$ is a finite sequence $c_0 a_1c_1  \cdots a_n c_n$, for some $n \geq 1$, such that: $(1)$  $c_0=c$ and $c_n=c'$, and $(2)$  $c_i \by{a_{i+1}}_{\m{T}} c_{i+1}$ for all $i \in [0,n[$. 
In this case, we say that  $\rho$ has length $n$  and is labelled by the word $a_1 a_2 \cdots a_n$.

Let     $c,c' \in C$ and  $u \in \Sigma^*$. We  write $c \,{\bby{n}{u}} {}_{\m{T}} \,c'$ if one of the following two cases holds: (1)  $n=0$, $c=c'$,  and $u=\epsilon$, and (2)   there is a run $\rho$ of length $n$ from $c$ to $c'$ labelled  by $u$.  We also write $c \, \lby{u}{}_{\m{T}}^*\, {c'}$   (resp. $c \, \lby{u}{}_{\m{T}}^+\, {c'}$) to denote   that $c \,{\bby{n}{u}} {}_{\m{T}} \,c'$    for some $n\geq 0$ (resp. $n>0$).

For every      $C_1, C_2 \subseteq C$, let    $\mathit{Traces}_{\m{T}}(C_1,C_2)=\{u\in \Sigma^* \,|\, \exists (c_1,c_2) \in C_1 \times C_2\,,\, c_1 \, \lby{u}{}_{\m{T}}^* \, {c_2}\}$ be   the set of sequences of actions generated by the  runs of  $\m{T}$ from a configuration  in $C_1$  to a configuration in $C_2$.

For every  $C' \subseteq C$,    let $Pre_{\m{T}}(C')=\{c \in C\,|\, \exists (c',a) \in C' \times \Sigma_{\epsilon}\,,\,  c \,\by{a}_{\m{T}} \,c'\}$ be the set of immediate predecessors of $C'$. Let  $Pre_{\m{T}}^*$  be  the reflexive-transitive closure of $Pre_{\m{T}}$, and  let $Pre_{\m{T}}^+=Pre_{\m{T}} \circ Pre_{\m{T}}^*$ where the operator  $\circ$ stands for the function composition.

\medskip

\noindent
{\bf Finite state automata:}
\label{sec.fsa}
A finite state automaton (FSA) is a tuple $\m{A}=(Q,\Sigma,\Delta,I,F)$ where: $(1)$
$Q$ is the  finite non-empty set of states, $(2)$ $\Sigma$ is the  finite  input  alphabet, $(3)$ $\Delta \subseteq (Q \times \Sigma_{\epsilon} \times Q)$ is the transition relation,  $(4)$
$I \subseteq Q$ is the  set of initial states, and $(5)$ $F \subseteq Q$ is the  set of final states. We represent  a transition $(q,a,q')$ in $\Delta$ by $q\by{a}_{\m{A}} q'$.   Moreover, if   $I'$ and $F'$ are two subsets of $Q$, then  we  use    $\m{A}(I',F')$ to denote the finite state automaton  defined  by  the tuple $(Q,\Sigma,\Delta,I',F')$.

 The size of  $\m{A}$ is defined by $|\m{A}|=(|Q|+|\Sigma|+|\Delta|)$. We use  $\m{T}({\m{A}})=(Q,\Sigma,\Delta)$ to denote   the  transition system associated with  $\m{A}$. The language  accepted (or recognized) by $\m{A}$ is   given by  $L(\m{A})=\mathit{Traces}_{\m{T}(\m{A})}(I,F)$.



\section{Generalized pushdown automata}
\label{chap2.gpa}
In this section, we  introduce  the class of generalized pushdown automata where transitions on stacks are (1) pop the top symbol of the stack, and (2) push a word in some (effectively) given set of words $L$ over the stack alphabet.  
A transition $t$ is of the form $\delta(p,\gamma,a,p')=L$ where $L$ is a (possibly infinite) set  of words. Being in a configuration $(q,w)$ where $q$ is a state and $w$ is a stack content, $t$ can be applied if both $p=q$ and the content of the stack is of the form $\gamma w'$ for some $w'$. Taking the transition and reading the input letter $a$ (which may be the empty word), the automaton moves to the successor configuration $(p',uw')$ where $u \in L$ (i.e., the new state is $p'$, and $\gamma$ is replaced with a word $u$ belonging to the language $L$).   Formally, we have:

\begin{defi}[Generalized pushdown automata]
A generalized pushdown automaton (GPA for short)  is a tuple $\m{P}=(P,\Sigma,\Gamma,\delta,p_0,\gamma_0,F)$ where: $(1)$  $P$ is the   finite non-empty set of  states, $(2)$ $\Sigma$ is the  input alphabet, $(3)$  $\Gamma$ is  the  stack alphabet, $(4)$ $\delta\,:\,  P \times \Gamma \times \Sigma_{\epsilon}\times P  \ra  \mathit{Lang}({\Gamma})$ is the transition function, $(5)$ $p_0 \in P$ is the initial state, $(6)$ $\gamma_0 \in \Gamma$ is the initial stack symbol,  and $(7)$  $F \subseteq P$ is the set of final states.

\end{defi}

Next, we define the effectiveness  property for generalized pushdown automata. Intuitively,  the  generalized pushdown automaton $\m{P}$ is said to be effective if for any possible pushed language $L$ by $\m{P}$ (i.e., $\delta(p,\gamma,a,p')=L$ for some $p, p' \in Q$, $\gamma \in \Gamma$, and $a \in \Sigma_\epsilon$), the  problem of checking the non-emptiness of the intersection of $L$ and any given regular language (i.e. accepted by a finite-state automaton) is decidable. 

\begin{defi}[Effectiveness Property]
\label{def-effec}
A GPA $\m{P}=(P,\Sigma,\Gamma,\delta,p_0,\gamma_0,F)$  is   effective if and only if for every finite state automaton $\m{A}$ over the alphabet  $\Gamma$, it is decidable  whether $L(\m{A}) \cap \delta(p,\gamma,a,p') \neq \emptyset$ for all  $p,p' \in P$, $\gamma \in \Gamma$, and $a \in \Sigma_{\epsilon}$.
\end{defi}

A configuration of a GPA $\m{P}=(P,\Sigma,\Gamma,\delta,p_0,\gamma_0,F)$ is a pair $(p,w)$ where $p \in P$ and $w \in \Gamma^*$. The set of all configurations of $\m{P}$ is denoted by $\mathit{Conf}({\m{P}})$.   Similarly to the case of pushdown automata \cite{BEM97}, we   use  the class of $\m{P}$-automata  as  finite symbolic representation of a set of configurations of  GPA. Formally,  a $\m{P}$-automaton is a FSA  $\m{A}=(Q_{\m{A}},\Gamma,\Delta_{\m{A}}, I_{\m{A}}, F_{\m{A}})$ such  that  $I_{\m{A}}=P$. We say that a configuration $(p,w)$ of $\m{P}$ is accepted (or recognized) by $\m{A}$ if $w \in L(\m{A}(\{p\},F_{\m{A}}))$. The set of all configurations recognized by $\m{A}$ is denoted by $L_{\m{P}}(\m{A})$. A set of configurations of $\m{P}$ is said to be recognizable if and only if  it is accepted by some $\m{P}$-automaton.

\begin{table}[h]

\begin{center}
\fbox{$\begin{array}{lllclllclll}
\multicolumn{11}{l}{\m{P}=(\{p_0,p_1,p_2,p_f\},\{a,b,c\},\{\bot, \gamma_0, \gamma_1,\gamma_2\},\delta,p_0,\bot,\{p_f\})}\;\;\; \;\;\; \; \;\;\; \; \;\;\; \; \;\;\; \; \;\;\;  \\[1ex]

\delta(p_0,\bot,\epsilon,p_2)& = & \{\gamma_2^i \gamma_1^i \gamma_0^i\bot \mid i \in \mathbb{N}\} & \;\;\;\;\;\;\;\; \;\;\;  \;\;\;\; \;\;\;  &\delta(p_2,\gamma_2,a,p_2)& = &\{\epsilon\} & ~ &  \\

\delta(p_2,\gamma_1,b,p_1) & = &  \{\epsilon\} & &
\delta(p_1,\gamma_1,b,p_1)& = & \{\epsilon\}& ~ & 
\\
\delta(p_1,\gamma_0,c,p_0)& = & \{\epsilon\}& &
\delta(p_0,\gamma_0,c,p_0)& =& \{\epsilon\}& ~ &  \\

\delta(p_0,\bot,\epsilon,p_f)& = & \{\epsilon\}& &
 \text{otherwise}& & \emptyset& ~ & \\

\end{array}$}

\caption{\label{exp2}\footnotesize A GPA $\m{P}$ for $\{\epsilon\} \cup \{a^{i_1}b^{i_1} c^{i_1} a^{i_2} b^{i_2}c^{i_2}\cdots a^{i_k}b^{i_k}c^{i_k}\;|\; k \geq 1$ and $i_1,\ldots,i_k > 0\} $}
\end{center}

\end{table}

The transition system $\m{T}(\m{P})$ associated with the generalized pushdown automaton $\m{P}$  is defined   by the tuple  $(\mathit{Conf}(\m{P}),\Sigma,\ra)$  where $\ra$ is the smallest transition relation such that:  For every  $p,p' \in P$, $\gamma \in \Gamma$, and $a \in \Sigma_{\epsilon}$,  if $\delta(p,\gamma,a,p')\neq \emptyset$, then $(p, \gamma w) \by{a}_{\m{T}(\m{P})} (p' ,u w)$ for all $u \in \delta(p,\gamma,a,p')$ and  $w \in \Gamma^*$.  Let  $L(\m{P})=\mathit{Traces}_{\m{T}(\m{P})}(\{(p_0, \gamma_0)\},F \times \{\epsilon\})$  denote the language accepted by   $\m{P}$.

Observe that   pushdown automata  can be seen as  a particular class  of effective GPA  where  $\delta(p,\gamma,a,p')$ is a finite set of words for all $(p,\gamma,a,p')$.

Table~\ref{exp2} shows an example of an effective  generalized pushdown automaton where the pushed language $\{\gamma_2^i \gamma_1^i \gamma_0^i\bot \mid i \in \mathbb{N}\} $ can be accepted by a Petri net (with reachability as  acceptance condition).

%
%
%

 \subsection{Computing the set of    predecessors for an GPA}
 \label{sat}
 In this section, we show that the set of predecessors of a recognizable set of configurations of   an effective GPA is recognizable and effectively constructible.
 This is done by adapting the  construction given in \cite{BEM97,EHRS00b,Sch02b}. 
  On the other hand, it is easy to observe that the set of successors  of a recognizable set of configurations of an effective GPA is not recognizable in general (see the example given in  Table~\ref{exp2}).

 \begin{thm}
 \label{pred-EPDA}
 For every effective generalized pushdown automaton $\m{P}$, and every $\m{P}$-automaton $\m{A}$, it is possible to construct  a $\m{P}$-automaton recognizing $Pre^*_{\m{T}(\m{P})}(L_{\m{P}}(\m{A}))$.
 \end{thm}
 
The rest of this section is devoted to the proof of  Theorem \ref{pred-EPDA}. For that, let  $\m{P}=(P,\Sigma,\Gamma,\delta,p_0,\gamma_0,F)$ be an effective generalized pushdown automata  and  $\m{A}=(Q_{\m{A}},\Gamma,\Delta_{\m{A}},I_{\m{A}},F_{\m{A}})$ be an $\m{P}$-automaton. Without loss of generality, we assume that $\m{A}$ has no transition leading to an initial state. We compute $Pre^*_{\m{T}(\m{P})}(L_{\m{P}}(\m{A}))$ as the set of configurations recognized by an $\m{P}$-automaton $\m{A}_{pre^*}=(Q_{\m{A}},\Gamma,\Delta_{pre^*},I_{\m{A}},F_{\m{A}})$ obtained from $\m{A}$ by means of a saturation procedure. Initially, we have $\m{A}_{pre^*}=\m{A}$. Then, the procedure adds new transitions to $\m{A}_{pre^*}$, but no new states. New transitions are added according to the following saturation rule:


\begin{table}[h]
\begin{center}

\begin{tabular}{|l|}
\hline
{\normalsize \em   For every $p,p' \in P$, $\gamma \in \Gamma$, and $a \in \Sigma_{\epsilon}$, if $\delta(p,\gamma,a,p') \neq \emptyset$, then for every  $q \in Q_{\m{A}}$ } \\  {\normalsize  \em such that $\delta(p,\gamma,a,p') \cap L(\m{A}_{pre^*}(\{p'\},\{q\})) \neq \emptyset$,  add the transition  $(p,\gamma,q)$ to $\m{A}_{pre^*}$}\\

\hline
\end{tabular}

\end{center}

\end{table}


It is easy to see that the saturation procedure eventually reaches a fixed point because the number of possible new transitions is finite.  Moreover, the saturation procedure is well defined since    the emptiness problem of   the language $\big(\delta(p,\gamma,a,p') \cap L(\m{A}_{pre^*}(\{p'\},\{q\}))\big)$  is decidable ($\m{P}$ is an effective GPA).  Then, the relation between  the set of configurations recognized by $\m{A}_{pre^*}$  and  the set $Pre^*_{\m{T}(\m{P})}(L_{\m{P}}(\m{A}))$ is established  by Lemma \ref{lemm_pre_epda}. (Observe that Theorem \ref{pred-EPDA} follows from Lemma  \ref{lemm_pre_epda}.)

\begin{lem}
\label{lemm_pre_epda}
$L_{\m{P}}(\m{A}_{pre^*})=Pre^*_{\m{T}(\m{P})}(L_{\m{P}}(\m{A}))$.
\end{lem}

 Lemma \ref{lemm_pre_epda} is an immediate consequence  of Lemma \ref{lemma1-gpa} and Lemma \ref{lemm-gpa}: Lemma \ref{lemma1-gpa} shows that  $Pre^*_{\m{T}(\m{P})} (L_{\m{P}}(\m{A})) \subseteq L_{\m{P}}(\m{A}_{pre^*})$ while Lemma \ref{lemm-gpa} establishes $L_{\m{P}}(\m{A}_{pre^*})  \subseteq Pre^*_{\m{T}(\m{P})} (L_{\m{P}}(\m{A}))$.

\begin{lem}
\label{lemma1-gpa}
For every configuration $(p',w') \in L_{\m{P}}(\m{A})$,  if $(p, w)\, \lby{\tau}{}_{\m{T}(\m{P})}^*\, (p',w')$  for some $\tau \in \Sigma^*$, then $(p,w)  \in L_{\m{P}}(\m{A}_{pre^*})$.
\end{lem}

\begin{proof}
Assume $(p,w) \,{{\bby{n}{\tau}}}{}_{\m{T}(\m{P})} \, (p',w')$. We proceed by induction on $n$.

\medskip

\noindent
{\bf Basis.} $n=0$. Then, $p=p'$ and $w'=w$. Since $(p',w') \in L_{\m{P}}(\m{A})$ and $L_{\m{P}}(\m{A}) \subseteq L_{\m{P}}(\m{A}_{pre^*})$, we have $(p,w)  \in L_{\m{P}}(\m{A}_{pre^*})$.

\medskip

\noindent
{\bf Step.} $n>0$. Then, there is a configuration $(p'',w'') \in \mathit{Conf}(\m{P})$ such  that:

$$ (p,w) \,{\by{a}}_{\m{T}(\m{P})} \, (p'',w'') \,{{\bby{n-1}{\tau'}}}{}_{\m{T}(\m{P})} \, (p',w')$$
\noindent
 for some $a \in \Sigma_{\epsilon}$ and  $\tau' \in \Sigma^*$ such that $\tau=a \tau'$.
 
 \medskip
 
 \noindent
 We apply the induction hypothesis to $(p'',w'') \, {{\bby{n-1}{\tau}}}{}_{\m{T}(\m{P})} \, (p',w')$, and we obtain: $$(p'',w'') \in L_{\m{P}}(\m{A}_{pre^*})$$
 
 \noindent
 Since $(p,w) \,{\by{a}}_{\m{T}(\m{P})} \, (p'',w'') $, there are $\gamma \in \Gamma$ and  $u,v \in \Gamma^*$  such that:
 
 \begin{center}
 $w=\gamma v$, $w''=u v$, and $u \in  \delta(p,\gamma,a,p'')$
 \end{center}
 
 \noindent
 Let $q$ be a state of $\m{A}_{pre^*}$ such that:

 \begin{center}
 $u \in L(\m{A}_{pre^*}(\{p''\},\{q\}))$ and $v \in L(\m{A}_{pre^*}(\{q\},F_{\m{A}}))$.
 \end{center}
 
Such a state $q$ exists since $uv  \in L(\m{A}_{pre^*}(\{p''\},F_{\m{A}}))$.
By the saturation rule, we have that  $(p,\gamma,q)$ is a transition of $\m{A}_{pre^*}$ since $u \in L(\m{A}_{pre^*}(\{p''\},\{q\})) \cap \delta(p,\gamma,a,p'')$. This implies that $w=\gamma v \in L(\m{A}_{pre^*}(\{p\},F_{\m{A}})) $ since $(p,\gamma,q) \in \Delta_{\m{A}_{pre^*}}$ and $v \in L(\m{A}_{pre^*}(\{q\},F_{\m{A}}))$. 
Hence, we have $ (p,w) \in L_{\m{P}} (\m{A}_{pre^*})$. \end{proof}

In the following, we  establish that $L_{\m{P}}(\m{A}_{pre^*})  \subseteq Pre^*_{\m{T}(\m{P})} (L_{\m{P}}(\m{A}))$. This is  an an immediate corollary of the following lemma:

\begin{lem}
\label{lemm-gpa}
If $w \in L(\m{A}_{pre^*}(\{p\},\{q\}))$, then $(p,w)\, \lby{\tau}{}_{\m{T}(\m{P})}^* \, (p',w')$ for a configuration $(p',w')$ and $\tau \in \Sigma^*$ such that  $w' \in L(\m{A}_0(\{p'\},\{q\}))$. Moreover,   if $q $ is an initial state of $\m{A}_{pre^*}$, then we have $p'=q$ and $w'=\epsilon$.\end{lem}

\begin{proof}
Let  $\m{A}_n=(Q_{\m{A}},\Gamma,\Delta_i,I_{\m{A}},F_{\m{A}})$ be  the $\m{P}$-automaton obtained   after adding $n$ transitions to $\m{A}$. In particular, we have $\m{A}_0=\m{A}$. Then,  it is easy to see   that $L_{\m{P}}(\m{A})=L_{\m{P}}(\m{A}_0) \subseteq L_{\m{P}}(\m{A}_1) \subseteq L_{\m{P}}(\m{A}_2)  \subseteq \cdots \subseteq L_{\m{P}}(\m{A}_{pre^*}) $.

Let $n$ be an index such that $w \in L(\m{A}_{n}(p,q))$ holds. We shall prove the first part of Lemma   \ref{lemm-gpa} by induction on $n$. The second part follows immediately from the fact that initial states have no incoming transitions in $\m{A}_0$.

\medskip

\noindent
{\bf Basis.} $n=0$. Since $w \in L(\m{A}_{n}(\{p\},\{q\}))$  holds, take $w'=w$ and $p'=p$.

\medskip

\noindent
{\bf Step.} $n>0$. Let $t=(p'',\gamma,q')$  be the $n$-th transition added to $\m{A}_{pre^*}$. Let $m$ be the number of times that $t$ is used in $p \,\lby{w}{}_{\m{T}(\m{A}_n)}^* \, q$. 

The proof is by induction on $m$. If $m=0$, then we have $w \in L(\m{A}_{n-1}(\{p\},\{q\}))$, and the property (1) follows from the induction hypothesis (induction on $n$). So,  assume that $m>0$. Then there exist $u$ and $v$ such that $w=u \gamma v$ and 

\begin{center}
$u \in L(\m{A}_{n-1}(\{p\},\{p''\}))$, $p'' \by{\gamma}_{\m{T}(\m{A}_n)} q'$, and $v \in L(\m{A}_{n}(\{q'\},\{q\}))$.
\end{center}

\noindent
The application of the induction hypothesis to $u \in L(\m{A}_{n-1}(\{p\},\{p''\}))$ yields to that:

\begin{center}
  $(p,u) \, \lby{\tau'}{}_{\m{T}(\m{P})}^* \, (p'',\epsilon)$ for some $\tau' \in \Sigma^*$.
\end{center}  

\medskip

\noindent
Since the transition $(p'',\gamma,q')$ has been added by applying the saturation procedure, there exist $p'''$,  $w_2$, and $a \in \Sigma_{\epsilon}$ such that:

\begin{center}
 $w_2 \in \delta(p'',\gamma,a,p''')$, and $w_2 \in L({\m{A}_{n-1}}(\{p'''\},\{q'\}))$.
\end{center}

\medskip

\noindent
From  $w_2 \in L({\m{A}_{n-1}}(\{p'''\},\{q'\}))$ and $v \in L(\m{A}_{n}(\{q'\},\{q\}))$, we get that there is a computation  $\rho=p''' \,\lby{w_2v}{}_{\m{T}(\m{A}_n)}^* \, q$ such that  the number of times transition $t$ is used is strictly less than $m$. So, we can apply the induction hypothesis (induction on $m$) to $\rho$, and we obtain:

\medskip

\noindent
$(p''',w_2v) \, \lby{\tau''}{}^*_{\m{T}(\m{P})} \, (p',w')$ for a configuration $(p',w')$ and $\tau'' \in \Sigma^*$ s.t. $w' \in L(\m{A}_{0}(\{p'\},\{q\}))$.

\medskip

\noindent
Putting all previous equations together, we get  $w' \in L(\m{A}_{0}(\{p'\},\{q\}))$, and:

\begin{center}
$ (p,w)=(p,u \gamma v)\,   \lby{\tau'}{}_{\m{T}(\m{P})}^* \, (p'',\gamma v)\,  \by{a}_{\m{T}(\m{P})} \, (p''',w_2 v) \,  \lby{\tau''}{}_{\m{T}(\m{P})}^* \,(p',w')$
\end{center}

\noindent
This terminates the proof of Lemma \ref{lemm-gpa}.\end{proof}


 \subsection{Emptiness problem and closure properties for GPA}
 \label{chap2.empt}
 In this section, we show that the emptiness problem  is decidable  for effective generalized pushdown automata. This is an immediate consequence  of the fact that the set of predecessors of a recognizable set of configurations of  an effective generalized pushdown automaton is also recognizable and effectively constructible.
 
 \begin{thm}
 \label{emptiness-EPDA}
 The emptiness problem is decidable for effective generalized pushdown automata.
 \end{thm}
 
 \begin{proof}
 Let $\m{P}=(P,\Sigma,\Gamma,\delta,p_0,\gamma_0,F)$ be an effective generalized pushdown automaton. It is easy to see that $L(\m{P}) \neq \emptyset$ if and only if $(p_0,\gamma_0) \in Pre_{\m{T}(\m{P})}^*(F \times \{\epsilon\})$. By Theorem \ref{pred-EPDA}, we can construct a $\m{P}$-automaton ${\m{A}_{pre^*}}$ that recognizes exactly the set $Pre_{\m{T}(\m{P})}^*(F \times \{\epsilon\})$ since $F \times \{\epsilon\}$ is a recognizable set of configurations. Hence, the emptiness problem for $\m{P}$ is  decidable since checking whether  $(p_0,\gamma_0)$ is in  $L_{\m{P}}(\m{A}_{pre^*})$ is decidable. \end{proof}

Next, we show some  closure properties for effective generalized pushdown automata. 
\label{closure-property}

 \begin{thm}
The class of effective GPAs  is  closed under  concatenation, union, Kleene star,  projection, homomorphism, and intersection with a regular language. However,   effective GPAs are not closed under  intersection.
 \end{thm}
 
  \begin{proof}
 Showing the  closure of the class of effective generalized pushdown automata  under  concatenation, union, Kleene star,  projection, homomorphism, and intersection with a regular language is similar to the case of pushdown automata.  The only issue to   prove  is the closure under intersection. For that, let us  assume by contradiction that the class of effective 
 GPAs is closed under the intersection operation. Let $\m{P}_1$ and $\m{P}_2$ two pushdown automata. Since the class of  pushdown automata is  a particular class of effective generalized pushdown automata and  the class of effective generalized pushdown automata  is closed under the intersection operation (from  the contradiction's  hypothesis),  there is an effective  generalized pushdown automaton  $\m{P}$  such that $L(\m{P})=L(\m{P}_1) \cap L(\m{P}_2)$. Applying  Theorem \ref{emptiness-EPDA} to  $\m{P}$, we obtain the decidability  of the emptiness problem  of the intersection of two context-free languages, which is a contradiction.
   \end{proof}


\section{Ordered multi-pushdown automata}
\label{chap2.model}

In this section, we first recall the  definition of \emph{multi-pushdown automata}. Then   \emph{ordered multi-pushdown automata}   appear as a special case of
{multi-pushdown automata}.

\subsection{Multi-pushdown automata}
Multi-pushdown automata have one read-only left to right input
tape and $n \geq 1$ read-write memory tapes (stacks) with a last-in-first-out
rewriting policy. A transition is of the form $t = \co q,\gamma_1,\ldots,\gamma_n\cf
\by{a} \co q',\alpha_1,\ldots,\alpha_n \cf$. Being in a configuration
$(p,w_1,\ldots,w_n)$, which is composed of a state
$p$ and a stack content $w_i$ for each  stack
$i$, $t$ can be applied if both $q = p$ and the $i$-th stack is of the form
$\gamma_i w_i'$ for some $w_i'$. Taking the transition and reading the input symbol 
$a$ (which might be the empty word), the automaton moves to the successor
configuration $({q',\alpha_1 w_1',\ldots,\alpha_n w_n'})$.

\begin{defi}[Multi-pushdown automata]
  A {\em multi-pushdown automaton} (MPA) is a tuple  $\m{M}=(n,Q,\Sigma,\Gamma,\Delta,q_0,\gamma_0, F)$ where: 
  
  \begin{iteMize}{$\bullet$}
  \item    $n \ge 1$ is the number of stacks.
  \item    $Q$ is the finite non-empty set of \emph{states}.
  \item 
$\Sigma$ is the finite set of \emph{input symbols}. 
\item 
$\Gamma$ is the finite set of \emph{stack symbols} containing the special stack symbol $\bot$. 

\item 
 $\Delta \subseteq \big(Q \times (\Gamma_{\epsilon})^n\big) \times \Sigma_{\epsilon} \times \big(Q
  \times (\Gamma^\ast )^n\big)$ is the
  \emph{transition  relation} such that, for all
  $((q,\gamma_1,\ldots,\gamma_n),a,(q',\alpha_1,\ldots,\alpha_n)) \in \Delta
 $ and $i \in [1,n]$,  we have: 
 
 \begin{iteMize}{$-$}
  \item $|\alpha_i| \leq 2$.
  \item If $\gamma_i \neq  \bot$, then  $\alpha_i  \in (\Gamma \setminus \{\bot\})^*$.
  \item   If $\gamma_i = \bot$, then  $\alpha_i =
  \alpha'_i \bot$ for some $\alpha'_i \in (\Gamma_{\epsilon} \setminus \{\bot\})$. 
  \end{iteMize}

\item      $q_0 \in Q$ is the \emph{initial state}.
\item 
 $\gamma_0 \in (\Gamma \setminus \{\bot\})$ is the \emph{initial stack symbol}.
 \item 
$F \subseteq Q$ is the set of \emph{final states}. 

 \end{iteMize} 
  \end{defi}

  The size of $\m{M}$, denoted by $|\m{M}|$, is defined  by $(n+|Q|+|\Sigma|+|\Gamma|+|\Delta|)$.   
In the rest of this paper,  we   use  $\co q,\gamma_1,\ldots,\gamma_n \cf \by{a}_{\m{M}}
 \co q',\alpha_1,\ldots,\alpha_n \cf$ to denote that the transition     $((q,\gamma_1,\ldots,\gamma_n),a,(q',\alpha_1,\ldots,\alpha_n))$  is in $\Delta$. Moreover, we denote by $\m{M}(q,\gamma,q')$ the multi-pushdown automaton  defined by the tuple $(n,Q,\Sigma,\Gamma,\Delta,q,\gamma,\{q'\})$.

 A {stack content} of $\m{M}$ is an element  of $\mathit{Stack}(\m{M}) = (\Gamma \setminus \{\bot\})^*
\{\bot\}$. A { configuration} of
$\m{M}$ is a $(n+1)$-tuple $({q,w_1,\ldots,w_n})$ with $q \in
Q$, and $w_1,\ldots,w_n \in \mathit{Stack}(\m{M})$. A configuration $(q,w_1,\ldots,w_n)$ is final  if  $q \in F$ and $w_1=\cdots=w_n=\bot$.  The
set of configurations of $\m{M}$ is denoted by $\mathit{Conf}{(\m{M})}$.

The behavior of  $\m{M}$  is
described by its corresponding transition system $\m{T}(\m{M})$  defined by  the tuple $(\mathit{Conf}(\m{M}),\Sigma,\ra)$ where $\ra$ is the smallest transition relation satisfying the following condition:  if $\co q,\gamma_1,\ldots,\gamma_n \cf \by{a}_{\m{M}} \co q',\alpha_1,\ldots,\alpha_n \cf$, then $(q,\gamma_1 w_1,\ldots,\gamma_n w_n) \by{a}_{\m{T}(\m{M})} (q',\alpha_1 w_1,\ldots,\alpha_n w_n)$ for all $w_1,\ldots,w_n \in \Gamma^*$ such that $\gamma_1 w_1,\ldots,\gamma_n w_n \in \mathit{Stack}(\m{M})$.  Observe that the symbol $\bot$ marks the bottom of a stack. According to the
transition relation, $\bot$ can never be popped. 

The language accepted (or recognized) by $\m{M}$ is defined by the set 
$L(\m{M})  = \{\tau \in \Sigma^* \mid  (q_0,\gamma_0\bot,\bot,\ldots,\bot)\,
  \lby{\tau}{}_{\m{T}(\m{M})}^*\, c\; \text{ for some final configuration } c\}$.

\subsection{Symbolic representation of MPA configurations}

We show in this section how we can symbolically represent infinite sets of multi-pushdown automaton configurations using special kind of finite automata which were introduced in \cite{anil10}.  Let  $\m{M}=(n,Q,\Sigma,\Gamma,\Delta,q_0,\gamma_0, F)$ be a multi-pushdown automaton. An $\m{M}$-automaton for accepting configurations of $\m{M}$ is a finite state automaton $\m{A}=(Q_{\m{M}},\Gamma,\Delta_{\m{M}},I_{\m{M}},F_{\m{M}})$ such that $I_{\m{M}}=Q$. We say that a configuration $(q,w_1,\ldots,w_n)$ of $\m{M}$ is accepted (or recognized) by $\m{A}$ if and only if the word  $w=w_1 w_2  \cdots  w_n $ is in $L(\m{A}(\{q\},F_{\m{M}}))$. (Notice that for every word $w \in L(\m{A}(\{q\},F_{\m{M}}))$ there are  unique words $w_1,\ldots,w_n \in \mathit{Stack}(\m{M})$ such that $w=w_1 \cdots w_n$.) 
The set of all configurations recognized by $\m{A}$ is denoted by $L_{\m{M}}(\m{A})$. A set of configurations of $\m{M}$ is said to be recognizable if and only if it is accepted by some $\m{M}$-automaton. Finally, it is easy to see that  the class of $\m{M}$-automata is closed under  all the boolean operations and that  emptiness and membership problems are decidable in polynomial time.

\subsection{Ordered multi-pushdown automata}
An ordered multi-pushdown automaton is a multi-pushdown automaton in which one can pop only from the first non-empty stack (i.e., all preceding stacks are  equal to $\bot$).

\begin{defi}[Ordered multi-pushdown automata]
An {ordered multi-pushdown automaton} (OMPA for short) is a multi-pushdown automaton 
  $(n,Q,\Sigma,\Gamma,\Delta,q_0,\gamma_0,F)$ where, for each transition  $\co q,\gamma_1,\ldots, \gamma_n \cf \by{a}_{\m{M}}\, \co q',\alpha_1,\ldots,\alpha_n\cf$, there is an $i \in [1,n]$ such that $\gamma_1=\cdots=\gamma_{i-1}=\bot$, $\gamma_i \in \Gamma_{\epsilon}$, and $\gamma_{i+1}= \cdots = \gamma_n= \epsilon$.

\end{defi}

\medskip

We introduce the following  abbreviations:  $(1)$ For $n \ge 1$, we call an 
MPA/OMPA an $n$-MPA/$n$-OMPA,
respectively, if its number of stacks is $n$, and $(2)$  An MPA \emph{over} $\Sigma$
is an MPA with input alphabet $\Sigma$.

\medskip

  In the following,  we consider only ordered multi-pushdown automata in some normal form. This normal form is used   only  to  simplify the presentation. (Observe that this form is slightly more general than the one considered   in \cite{ABH-dlt08}.) In such normal form, any transition, that pops a symbol from the $i^{th}$ stack with $i \in [2,n]$, is only allowed to push a symbol on the first stack. Furthermore, pushing symbols on the stacks from $1$ to $n$ is only allowed  while popping a symbol from the first stack.

\begin{defi}
An OMPA 
  $(n,Q,\Sigma,\Gamma,\Delta,q_0,\gamma_0,F)$ is in normal form if   $\Delta$ contains only the following types of transitions:

\medskip
  
\begin{iteMize}{$\bullet$}

\item $\co q,\gamma,\epsilon,\ldots,\epsilon \cf \by{a}_{\m{M}}
\, \co q',\alpha_1,\ldots,\alpha_n\cf$ for some  $q,q' \in Q$, $\gamma \in \Gamma $, $a \in \Sigma_{\epsilon}$,  $\alpha_1 \in \Gamma^*$ and $\alpha_j \in (\Gamma_\epsilon \setminus \{\bot\})$  for all $j \in [2,n]$. This transition pops a symbol from the first stack  while pushing at most two symbols on the first stack and  at most one  symbol on the stacks from $2$ to $n$.

%
%
%
%
%
\medskip

\item $\co q,\bot,\ldots,\bot,\gamma,\epsilon,\ldots,\epsilon \cf \by{a}_{\m{M}}
\co q',\gamma'\bot,\bot,\ldots,\bot,\epsilon,\epsilon,\ldots,\epsilon\cf$ for some $q,q' \in Q$, $\gamma,\gamma' \in (\Gamma \setminus \{\bot\})$ and $a \in \Sigma_{\epsilon}$. This transition pops the stack symbol $\gamma$  from one of the stacks from $2$ to $n$ and pushes the stack symbol $\gamma'$ on the first stack.
%
%
%
%
\end{iteMize}
\end{defi}

 \medskip
 
\noindent We can show the equivalence (with respect to language acceptance) between the class of OMPA and OMPA in the normal form.

 \begin{lem}
 An $n$-OMPA $\m{M}$ can be transformed into an $n$-OMPA $\m{M}'$ in normal form with linear blowup in its size such that $L(\m{M})=L(\m{M}')$.
 \end{lem}

\begin{proof}
An easy generalization of the proof for the Chomsky normal form for
context-free grammars.
\end{proof}

{\em In the rest of the paper, we  assume  that any  OMPA  is  in the normal form.} Next, we recall some properties of the class of languages recognized by
$n$-OMPA.

\begin{lem}[\cite{multi96}]
\label{closurepropOMPA3}
 If   $\m{M}_1$  and $\m{M}_2$ are  two $n$-OMPAs over an  alphabet $\Sigma$, then it is possible to construct an $n$-OMPA $\m{M}$  over $\Sigma$ such that: $(1)$ $L(\m{M})=L(\m{M}_1) \cup L(\m{M}_2)$ and   $|\m{M}|=O(|\m{M}_1|+ |\m{M}_2|)$.
\end{lem}

\begin{lem}[\cite{multi96}]
\label{closurepropOMPA2}
 Let $\Sigma$ be an alphabet. Given an $n$-OMPA $\m{M}$ over $\Sigma$ and a finite state automaton $\m{A}$ over $\Sigma$, then it is possible  to construct an $n$-OMPA 
$\m{M}'$ such that:  $L(\m{M}')=L(\m{M}) \cap L(\m{A})$ and $|\m{M}'|=O(|\m{M}| \cdot |{\m{A}}|)$.

\end{lem}



\section{The emptiness problem for  a $n$-OMPA is in 2ETIME}
\label{chap2-dir1-2ETIME}
In this section, we show that the emptiness problem for ordered pushdown automata is in 2ETIME.  (We provide here a simpler proof of the 2ETIME upper bound than the one given in \cite{ABH-dlt08}.) To this aim, we  show that, given an OMPA $\m{M}$ with $n>1$ stacks, it is possible  to construct an effective generalized pushdown automaton   $\m{P}$, whose  pushed languages are defined by   OMPA with $(n-1)$ stacks of size $O(|\m{M}|^2)$, such that  the following invariant is preserved:  The state  and the stack content of $\m{P}$ are respectively the same as the state and the content of the $n^{th}$ stack of $\m{M}$ when its first $(n-1)$ stacks are empty (and so, $L(\m{P}) \neq \emptyset$ if and only if $L(\m{M})\neq \emptyset$).  Let $C$ be a recognizable set of configurations of $\m{M}$, and $\mathit{Pre}_{\m{T}(\m{M})}^*(C)$ the set of predecessors of $C$.  
Then, we can apply the saturation procedure  to $\m{P}$ to show that  the set of configurations  $C_n$,  consisting of   $\mathit{Pre}_{\m{T}(\m{M})}^*(C)$  restricted to the configurations in which   the first $(n-1)$ empty stacks are empty, is   recognizable and effectively constructible.  To compute the intermediary configurations in  $\mathit{Pre}_{\m{T}(\m{M})}^*(C)$ where  the first $(n-1)$ stacks   are not empty, we construct an ordered multi-pushdown automaton  $\m{M}'$ with $(n-1)$ stacks that: $(1)$ performs  the same transitions on its stacks as the ones   performed by $\m{M}$ on its first $(n-1)$ stacks, and $(2)$ simulates a push transition  of $\m{M}$ over  its  $n^{th}$ stack  by a transition of  the finite-state automaton accepting  the recognizable  set of configurations $C_n$. Now, we can apply the induction hypothesis to $\m{M}'$ and construct a finite-state automaton accepting the set of all  predecessors $\mathit{Pre}_{\m{T}(\m{M})}^*(C)$. 

Then, we  prove, by induction on $n$, that the emptiness problem for  the  $n$-OMPA $\m{M}$ is in 2ETIME with respect to the number of stacks.  For that, we assume that the emptiness problem  for  $(n-1)$-OMPA can be solved in  2ETIME. This implies that the generalized pushdown automaton $\m{P}$ (that simulates $\m{M}$) is effective (see Definition \ref{def-effec} and Lemma \ref{closurepropOMPA2}). Now, we can use  Theorem \ref{emptiness-EPDA} to prove   the decidability of the emptiness problem of the effective generalized pushdown automaton $\m{P}$ (and so,   of  the $n$-OMPA $\m{M}$). To show that the emptiness problem of $\m{P}$ and $\m{M}$ is in 2ETIME,  we  estimate the running time of  our saturation procedure, given in section \ref{sat}, under the assumption that the emptiness problem  for  $(n-1)$-OMPA can be solved in  2ETIME.

Let us give in more details  of the proof  described above.

\subsection{Simulation of an OMPA  by  an GPA}
In the following, we prove that, given an OMPA  $\m{M}$, we can construct a GPA   $\m{P}$,  with   transition languages defined by    $(n-1)$-OMPAs  of   size  $O(|\m{M}|^2)$, such that the emptiness problem for $\m{M}$ is reducible to the emptiness problem for $\m{P}$. (Recall that any OMPA is assumed to be in  normal form.)

\begin{thm}
\label{ompa-gpa}
Given an OMPA $\m{M}=(n,Q,\Sigma,\Gamma,\Delta,q_0,\gamma_0,F)$     with $n>1$,  it is possible to construct an GPA  $\m{P}=(P,\Sigma',\Gamma,\delta,p_0,\bot,\{p_f\})$  such that $P=Q \cup \{p_0,p_f\}$, $\Sigma'=\emptyset$, and  we have:

\begin{iteMize}{$\bullet$}

\item $L(\m{M})\neq \emptyset$ if and only if $L(\m{P}) \neq \emptyset$, and

\smallskip

\item For every $p_1,p_2 \in P$ and $\gamma \in \Gamma$,  there is an $(n-1)$-OMPA $\m{M}_{(p_1,\gamma,p_2)}$ over $\Gamma$  such that $L(\m{M}_{(p_1,\gamma,p_2)})=\big(\delta(p_1,\gamma,\epsilon,p_2)\big)^R$  and  $|\m{M}_{(p_1,\gamma,p_2)}|=O( |\m{M}|^2)$.

\end{iteMize}

\end{thm}

\noindent The remaining part of this subsection is devoted to the proof of Theorem \ref{ompa-gpa}. 
Let us present the main steps of the construction of the generalized pushdown automaton $\m{P}$. For that, let us consider an accepting  run   $\rho$ of $\m{M}$. This  run    can be seen  as   a sequence of runs  of the form $\varsigma_1 \sigma_1 \varsigma_2 \sigma_2 \cdots \varsigma_m \sigma_m$   such that  the pop  transitions  operations are exclusive to  the first $(n-1)$-stacks   (resp. the $n^{th}$ stack) of $\m{M}$  along  the sequence of runs $\varsigma_1,\varsigma_2, \ldots, \varsigma_m$ (resp.  $\sigma_1,\sigma_2, \ldots,\sigma_m$). Observe that, by definition,    the first $(n-1)$-stacks  of $\m{M}$ are  empty  along the runs $\sigma_1,\sigma_2, \ldots,\sigma_m$.  Moreover, at the beginning   of  the runs $\varsigma_1,\varsigma_2, \ldots, \varsigma_m$,  the OMPA $\m{M}$ is in  some configuration $c$  where the first stack of $\m{M}$ contains just one symbol and the stacks from 2 to $n-1$ are empty (i.e., $c$ of the form $(q,\gamma \bot,\bot,\ldots,\bot,w)$). Observe that this is an immediate consequence of the normal form that we have considered (since  $\m{M}$ is  only allowed to push just one symbol on the first stack while popping a symbol from the stacks from $2$ to $n$). In the case that $\m{M}$ is not in the normal form, notice that the set of all possible contents of the first $(n-1)$-stacks, at the beginning   of  the runs $\varsigma_1,\varsigma_2, \ldots, \varsigma_m$,  is still  finite. Later, we will use  this observation to show how we can adapt our construction to the  general case (when $\m{M}$ is not in the normal form).

\medskip

 Then,  we construct   $\m{P}$ such that  the following invariant is preserved during the simulation of $\m{M}$: The  state and the  content  of the stack of $\m{P}$ are  the same     as the state and the  content of the $n^{th}$ stack of  $\m{M}$ when its  first $(n-1)$-stacks are empty (and so, $L(\m{P}) \neq \emptyset$ if and only if $L(\m{M})\neq \emptyset$). To this aim, a pushdown   transtion  of $\m{M}$ that pops a symbol $\gamma$  from  its $n^{th}$ stack  is  simply simulated by  a  pushdown  transition of $\m{P}$ that pops the same symbol  $\gamma$. This implies that a run of the form $\sigma_i$, with $1\leq i\leq m$, that pops the word  $u_i$ from the $n^{th}$ stack of $\m{M}$ is simulated  by a run  of $\m{P}$ that pops  the same word $u_i$.  Now,  for every $j \in [1,m]$, we need to  compute     the pushed word  $v_j$  into  the  $n$-th stack of $\m{M}$    during the  run $\varsigma_j $ in order to be  pushed also by $\m{P}$. For that,   let $L_{(q,\gamma,q')}$ be the set of all possible pushed words $u$ into  the $n^{th}$ stack of $\m{M}$  by a run of the form  $(q,\gamma\bot,\bot,\ldots,\bot,w)\,  \lby{\tau}{}_{\m{T}(\m{M})}^* \, (q',\bot,\bot,\ldots,\bot,uw)$ where pop transitions are exclusive to  the first $(n-1)$-stacks  of $\m{M}$.  We show that this language $L_{(q,\gamma,q')}$ can be  characterized  by  an $(n-1)$-OMPA $\m{M}'(q,\gamma,q')$ over the stack alphabet  of $\m{M}$  that: $(1)$ performs the same transitions on its state and  $(n-1)$-stacks  as the one performed by $\m{M}$ on its state and its first $(n-1)$ stacks while discarding the pop transitions  of $\m{M}$ over the  $n^{th}$ stack, and $(2)$ makes visible as transition labels the pushed symbols over the $n^{th}$ stack of $\m{M}$. Now, to simulate the run $\varsigma_j=(q,\gamma\bot,\bot,\ldots,\bot,w)\,  \lby{\tau_j}{}_{\m{T}(\m{M})}^* \, (q',\bot,\bot,\ldots,\bot,uw)$  of $\m{M}$  ( which is equivalent to say that $\varsigma_j=(q,\gamma\bot,\bot,\ldots,\bot,\bot)\,  \lby{\tau_j}{}_{\m{T}(\m{M})}^* \, (q',\bot,\bot,\ldots,\bot,u)$),  $\m{P}$  can  push  into its stack the word $u$ such that $u^R \in L({\m{M}'}(q,\gamma,q'))$. If $\m{M}$ is not in the normal form, the run $\varsigma_j$ will be of the form $(q,\alpha_1,\ldots,\alpha_{n-1},w)\,  \lby{\tau_j}{}_{\m{T}(\m{M})}^* \, (q',\bot,\bot,\ldots,\bot,uw)$ (i.e., $(q,\alpha_1,\ldots,\alpha_{n-1},\bot)\,  \lby{\tau_j}{}_{\m{T}(\m{M})}^* \, (q',\bot,\bot,\ldots,\bot,u)$). In this case,  we can construct an $(n-1)$-OMPA $M_{(q',\alpha_1, \ldots,\alpha_{n-1},q')}$, which is precisely $\m{M}'(q,\gamma,q')$ with $(q,\alpha_1,\ldots,\alpha_{n-1})$ as initial configuration, characterizing the set of all possible pushed words $u$ on the $n^{th}$-stack. Thus,  the $(n-1)$-OMPA, occurring in Theorem \ref{ompa-gpa}, will be indexed by tuples of the form $(p_1,\alpha_1, \ldots,\alpha_{n-1},p_2)$ where $p_1, p_2 \in P$ and $\alpha_1, \ldots,\alpha_{n-1} \in (\Gamma_{\epsilon} \cup \Gamma^2)$.

The proof of Theorem \ref{ompa-gpa} will be structured as follows. First, we define an $(n-1)$-OMPA $\m{M}'$ over the alphabet $\Gamma$ that: $(1)$ performs the same transitions on its state and  $(n-1)$-stacks  as the one performed by $\m{M}$ on its state and its first $(n-1)$ stacks while discarding the pop transitions  of $\m{M}$  on  the  $n^{th}$ stack, and $(2)$ makes visible as transition labels the pushed symbols over the $n^{th}$ stack of $\m{M}$. Intuitively,  depending on the initial and final configurations of $\m{M}'$, the ``language'' of $\m{M}'$ summarizes the effect of a sequence of  pop transitions of $\m{M}$  over the first $(n-1)$-stacks   on the   $n^{th}$ stack of $\m{M}$.  So, if we are interested only by the configurations of $\m{M}$ where the first $(n-1)$ stacks are empty, a run of $\m{M}$ can be seen as a sequence of alternations of a pop transition of $\m{M}$ over the $n^{th}$ stack and a push  operation   over the $n^{th}$ stack of a word in  the ``language'' of  $\m{M}'$.

 Then, we construct a generalized pushdown automaton $\m{P}$ such that the state and  the stack content of $\m{P}$ are the same as the state and the $n^{th}$-stack content of $\m{M}$ when the first $(n-1)$ stacks of $\m{M}$ are empty. In the definition  of $\m{P}$, we use  the $(n-1)$-OMPA $\m{M'}$ to  characterize the pushed word  on the $n^{th}$ stack of $\m{M}$  due to   a sequence of  pop transitions of $\m{M}$  on  the $(n-1)$ first stacks of $\m{M}$. This implies that the emptiness problem for $\m{M}$ is reducible to its corresponding problem for $\m{P}$.


\paragraph{\bf Constructing the $(n-1)$-OMPA $\m{M}'$:}
\label{sect.ompan-1}
Let us introduce  the following   $n$-OMPA $\m{M}_{[1,n[}=(n,Q,\Sigma,\Gamma,\Delta_{[1,n[},q_0,\gamma_0,F)$   such that $\Delta_{[1,n[}=\Delta \cap \big(Q \times  (\Gamma_{\epsilon})^{n-1} \times \{\epsilon\}) \times \Sigma_{\epsilon} \times (Q \times (\Gamma^*)^{n}) \big)$. Intuitively, $\m{M}_{[1,n[}$  is built up  from $\m{M}$ by discarding pop transitions of $\m{M}$ over the $n^{th}$ stack. Then, let $\m{M'}=(n-1,Q,\Gamma,\Gamma,\Delta',q_0,\gamma_0,F)$  be the $(n-1)$-OMPA,  built out from  $\m{M}_{[1,n[}$, which $(1)$ performs the same transitions on the first ($n-1)$ stacks of $\m{M}_{[1,n[}$, and $(2)$ makes visible as transition label the pushed stack  symbol over the $n^{th}$ stack of $\m{M}_{[1,n[}$. Formally,  $\Delta'$ is defined as  the smallest transition relation satisfying the following condition:

\medskip

\begin{iteMize}{$\bullet$}
\item  If  $\co q,\gamma_1,\ldots,\gamma_{n-1},\epsilon \cf \by{a}_{\m{M}_{[1,n[}}  \co q',\alpha_1,\ldots,\alpha_{n-1},\alpha_n \cf$ for some  $q,q' \in Q$, $\gamma_1,\ldots,\gamma_{n-1} \in \Gamma_{\epsilon}$, $a \in \Sigma_{\epsilon}$, and $\alpha_1,\ldots,\alpha_{n} \in \Gamma^*$, then $ \co q,\gamma_1,\ldots,\gamma_{n-1} \cf\by{\alpha_n}_{\m{M}'} \co q',\alpha_1,\ldots,\alpha_{n-1}\cf$.


\end{iteMize}

\medskip

\noindent Observe that the pushed word $\alpha_n$ over the $n^{th}$-stack consists in at most one symbol (i.e., $\alpha_n \in \Gamma_{\epsilon}$). If $\m{M}$ is not in the normal formal, $\alpha_n$ can be  of size two (say $\alpha_n=\gamma \gamma'$), and in this case we need to associate two transitions to $\m{M}'$ which first read the symbol $\gamma'$ and then the symbol $\gamma$.

\medskip

Let us now give the relation between  the effect of a sequence of operations   of $\m{M}_{[1,n[}$ on the   $n^{th}$-stack  and the language of $\m{M}'$. 
\begin{lem}
\label{lemm-rest-n-1}
For every $q,q' \in Q$,  and $w_1, w'_1,\ldots, w_n, w'_n \in \mathit{Stack}(\m{M}_{[1,n[})$, $(q,w_1,\ldots,w_n) $ $\lby{\tau}{}_{\m{T}(\m{M}_{[1,n[})}^* $ $(q',w'_1,\ldots,w'_n)$ for some $\tau \in \Sigma^*$ if and only if there is $u \in \Gamma^*$  such that $(q,w_1,\ldots,w_{n-1})$ $ \lby{u}{}_{\m{T}(\m{M}')}^* $ $(q',w'_1,\ldots,w'_{n-1})$ and $w'_n=u^Rw_n$.
\end{lem}

\begin{proof}
To prove Lemma \ref{lemm-rest-n-1}, it suffices to   observe that the following holds: For every $q,q' \in Q$, $w_1,\ldots,w_n \in \mathit{Stack}(\m{M}_{[1,n[})$, and $w'_1,\ldots,w'_n \in \mathit{Stack}(\m{M}_{[1,n[})$,  $(q,w_1,\ldots,w_n)  \by{a}_{\m{T}(\m{M}_{[1,n[})} (q',w'_1,\ldots,w'_n)$ for some $a \in \Sigma_{\epsilon}$ if and only if there is $b \in \Gamma_{\epsilon}$  such that $(q,w_1,\ldots,w_{n-1})$ $ \by{b}_{\m{T}(\m{M}')} $ $(q',w'_1,\ldots,w'_{n-1})$ and $w'_n=bw_n$. This observation can be established easily using the definition of $\m{M}'$ and $\m{M}_{[1,n[}$.
 \end{proof}

\paragraph{\bf Constructing the GPA $\m{P}$:}

We are ready now  to define  the generalized pushdown automaton   $\m{P}=(P,\Sigma',\Gamma,\delta,p_0,\bot,\{p_f\})$, with $P=Q \cup \{p_0,p_f\}$ and $\Sigma'=\emptyset$, that  keeps track of the state and   content of  the   $n^{th}$ stack of $\m{M}$ when the first $(n-1)$ stacks are empty.
Formally, $\m{P}$ is built from $\m{M}$   as follows: For every $p,p'\in P$ and $\gamma \in \Gamma$,  we have:

\medskip

\begin{iteMize}{$\bullet$}
\item If $p=p_0$, $\gamma=\bot$, and $p' \in Q$,  then  $\delta(p,\gamma,\epsilon,p')=\{u^R \bot \mid u \in L(\m{M}'(q_0,\gamma_0,p')) \}$.

\smallskip

\item  If  $p \in F$, $\gamma=\bot$, and $p'=p_f$, then $\delta(p,\gamma,\epsilon,p')=\{\epsilon\}$.

\smallskip

\item  If $p,p' \in Q$ and   $\gamma \neq \bot$ then $\delta(p,\gamma,\epsilon,p')=\bigcup_{(q,\gamma') \in \Xi} \big(L(\m{M}'(q,\gamma',p'))\big)^R$ where $\Xi=\{(q,\gamma') \in (Q \times \Gamma)  \mid \exists a \in \Sigma_{\epsilon},\,\co p,\bot,\ldots,\bot,\gamma \cf \, \by{a}_{\m{M}} \, \co q,\gamma'\bot,\bot,\ldots,\bot,\epsilon\cf\}$.

\smallskip

\item Otherwise, $\delta(p,\gamma,\epsilon,p')=\emptyset$.
\end{iteMize}

\medskip

Observe that  for every $p_1,p_2 \in P$,  and $\gamma \in \Gamma$, we can construct an $(n-1)$-OMPA $\m{M}_{(p_1,\gamma,p_2)}$ over $\Gamma$  such that $L(M_{(p_1,\gamma,p_2)})=\big(\delta(p_1,\gamma,\epsilon,p_2)\big)^R$  and  $|\m{M}_{(p_1,\gamma,p_2)}|=O( |\m{M}|^2)$. This can be easily proved using Lemma \ref{closurepropOMPA3}. 

\medskip

To complete the proof of Theorem \ref{ompa-gpa}, it remains to show   that the emptiness problem for $\m{M}$  is reducible to its corresponding problem for $\m{P}$. This is stated by Lemma \ref{lemm.2.4-main-empty}.
\begin{lem}
\label{lemm.2.4-main-empty}
$L(\m{M}) \neq \emptyset$ if and only if $L(\m{P})\neq \emptyset$.
\end{lem}

\begin{proof}

To prove   that $L(\m{M}) \neq \emptyset$ iff $L(\m{P}) \neq \emptyset$, we will show that the following invariant is preserved: The state and  content  of $\m{P}$ are the same as the state and content of the last stack of $\m{M}$ when its first $(n-1)$-stacks are empty. Thus, we will  split  the run  of $\m{M}$   at the transitions that pop some symbol from the last stack. This implies that  a run of $\m{M}$ can be decomposed as follows: $(1)$ First a   run from the initial configuration to the first reachable configuration where the first $(n-1)$-stacks are empty, and $(2)$ a sequence of runs  that are starting from a configuration where the first $(n-1)$-stacks are empty to the first reachable  configuration with empty first $(n-1)$-stacks. Lemma \ref{lemm.2.gpa->ompa}  and Lemma \ref{lemm.2.intial.ompa=gpa} establish the relation between these two kind of runs of $\m{M}$ and the runs of $\m{P}$. Lemma \ref{lemm.2.gpa->ompa} shows that   $\m{M}$ can move from a configuration of the form $(q,\bot,\ldots,\bot,w) $ to  a  configuration of the form $(q',\bot,\ldots,\bot,w') $  if and only if  $\m{P}$ can move from the configuration $(q,w)$ to the configuration $(q',w')$. Lemma \ref{lemm.2.intial.ompa=gpa} proves that $\m{P}$ can move, in one step, from the initial configuration $(p_0,\bot) $ to the configuration $(q,w)$ if and only if  $\m{M}_{[1,n[}$ can move from the initial configuration $(q_0,\gamma_0 \bot,\bot,\ldots,\bot)$ to the  configuration  $(q,\bot,\bot,\ldots,\bot,w)$. As an immediate consequence of Lemma \ref{lemm.2.gpa->ompa}  and Lemma \ref{lemm.2.intial.ompa=gpa}, we obtain    that $L(\m{M}) \neq \emptyset$ if and only if $L(\m{P}) \neq \emptyset$.

\medskip

\begin{lem}
\label{lemm.2.gpa->ompa}
For every  $q, q' \in Q$ and  $w,w' \in \mathit{Stack}(\m{M})$,  $(q,w) \lby{\epsilon}{}_{\m{T}(\m{P})}^* (q',w')$ if and only if  $(q,\bot,\ldots,\bot,w) $ $\lby{\tau}{}_{\m{T}(\m{M})}^* $ $(q',\bot,\ldots,\bot,w')$ for some $\tau \in \Sigma^*$.
\end{lem}

\medskip

\begin{proof}
({\bf The Only if direction}) In the following, we show that  if $(q,w) \lby{\epsilon}{}_{\m{T}(\m{P})}^* (q',w')$ then $(q,\bot,\ldots,\bot,w) $ $\lby{\tau}{}_{\m{T}(\m{M})}^* $ $(q',\bot,\ldots,\bot,w')$ for some $\tau \in \Sigma^*$. Assume that $(q,w)\, {\bby{i}{\epsilon}}{}_{\m{T}(\m{P})}\, (q',w')$. We proceed by induction on $i$.

\medskip
\noindent
{\bf Basis.} $i=0$. Then $q=q'$ and  $w=w'$. This implies that $(q,\bot,\ldots,\bot,w) $ $\lby{\tau}{}_{\m{T}(\m{M})}^* $ $(q',\bot,\ldots,\bot,w')$ holds with  $\tau=\epsilon$.

\medskip
\noindent
{\bf Step.} $i>0$. Then there is a configuration    $(p,v) \in \mathit{Conf}(\m{P})$ such that:

\begin{equation}
(q,w)\, {\bby{i-1}{\epsilon}}{}_{\m{T}(\m{P})}\, (p,v)\, \by{\epsilon}{}_{\m{T}(\m{P})} \,(q',w') \end{equation}

\medskip
\noindent
From the definition of $\m{P}$, we can  show  that  $p \in Q$ and $v \in \mathit{Stack}({\m{M}})$ since $q,q' \in Q$ and  $w,w' \in \mathit{Stack}(\m{M})$. Thus, 
we can apply the induction hypothesis to $(q,w)\, {\bby{i-1}{\epsilon}}{}_{\m{T}(\m{P})}\, (p,v)$, and we obtain:

\begin{equation}
\label{eq.2.1}
(q,\bot,\ldots,\bot,w)\, \lby{\tau'}{}_{\m{T}(\m{M})}^*\, (p,\bot,\ldots,\bot,v)\;\;\;\;\; \text{for some}  \;\; \tau' \in \Sigma^*
\end{equation}

\medskip
\noindent
Since $(p,v)\, \by{\epsilon}{}_{\m{T}(\m{P})} \,(q',w')$, there are $\gamma \in \Gamma$ and  $u,v' \in \Gamma^*$ such that:

\begin{equation}
v=\gamma v' ,\;\;\;\; w'=uv', \;\;\;\; \text{and}\;\; u \in \delta(p,\gamma,\epsilon,q')
\end{equation}

\medskip
\noindent
Using the definition of $\m{P}$,  we can show that  there are  $q'' \in Q$, $b \in \Sigma_{\epsilon}$, and   $\gamma' \in \Gamma$ such that:

\begin{equation}
\label{eq.2.2}
\co p,\bot,\ldots,\bot,\gamma \cf \by{b}{}_{\m{M}} \co q'',\gamma'\bot,\ldots,\bot,\epsilon\cf \;\; \text{and} \;\;u^R \in L(\m{M'}(q'',\gamma',q'))
\end{equation}

\medskip
\noindent
Since $u^R \in L(\m{M'}(q'',\gamma',q'))$,  we obtain:

\begin{equation}
\label{eq.11}
(q'',\gamma' \bot,\bot,  \ldots, \bot) \, \lby{u^R}{}_{\m{T}(\m{M'})}^* \,(q',\bot,\bot,\ldots,\bot)
\end{equation}

\medskip
\noindent
Now, we can apply Lemma \ref{lemm-rest-n-1}, to the computation of $\m{M'}$ given in   Equation \ref{eq.11} and the stack content $v'$, and we obtain:

\begin{equation}
\label{eq.2.3}
(q'',\gamma'\bot,\bot,\ldots,\bot,v') \,\lby{\tau''}{}_{\m{T}(M_{[1,n[})}^*\, (q',\bot,\bot,\ldots,\bot,uv') 
\end{equation}

\medskip
\noindent
Since $\co p,\bot,\ldots,\bot,\gamma\cf \by{b}{}_{\m{M}} \co q'',\gamma'\bot,\ldots,\bot,\epsilon\cf$ and  $v=\gamma v'$, we obtain the following computation of $\m{M}$:

\begin{equation}
\label{eq.2.4}
(p,\bot,\bot,\ldots,\bot,v) \by{b}{}_{\m{T}(\m{M})} (q'',\gamma'\bot,\bot,\ldots,\bot,v')
\end{equation}

\medskip
\noindent
Putting together  Equation \ref{eq.2.1}, Equation \ref{eq.2.3}, and Equation  \ref{eq.2.4}, we obtain:

\begin{equation}
(q,\bot,\ldots,\bot,w)\, \lby{\tau}{}_{\m{T}(\m{M})}^*\, (q',\bot,\ldots,\bot,w') \;\; \text{with} \;\; \tau=\tau' b \tau''
\end{equation}

This terminates the Only if direction of Lemma  \ref{lemm.2.gpa->ompa}.

\medskip
\noindent
({\bf The If direction}) In the following, we show that  if $(q,\bot,\ldots,\bot,w)$ $\lby{\tau}{}_{\m{T}(\m{M})}^*$ $ (q',\bot,\ldots,\bot,w')$ for some $\tau \in \Sigma^*$, then $(q,w) $ $\lby{\epsilon}{}_{\m{T}(\m{P})}^*$ $(q',w')$. Let us  assume that $\rho=(q,\bot,\ldots,\bot,w)$ $\lby{\tau}{}_{\m{T}(\m{M})}^*$ $ (q',\bot,\ldots,\bot,w')$ for some $\tau \in \Sigma^*$. The proof is by induction on the number of times that a pop transition over the $n^{th}$ stack of $\m{M}$ is used in   the run $\rho$.  Let $m$ be the number of times that a transition in $\Delta_n=(\Delta \setminus \Delta_{[1,n[})$ is used in  $\rho$.

\medskip
\noindent
{\bf Basis.} $m=0$. Then, $q=q'$, $w=w'$, and $\tau=\epsilon$ since no transitions from $\Delta_{n}$ are used in  $\rho$, and no transitions from $\Delta_{[1,n[}$ can be applied along the run $\rho$. This implies that $(q,w)\, \lby{\epsilon}{}_{\m{T}(\m{P})}^*\, (q',w')$ holds.

\medskip
\noindent
{\bf Step.} $m>0$. Then, there are $\gamma, \gamma' \in (\Gamma \setminus \{\bot\})$, $v \in \Gamma^*$, and $q_1,q_2 \in Q$  such that:

\begin{equation}
\rho_1=(q,\bot,\ldots,\bot,w) \, \lby{\tau'}{}_{\m{T}(\m{M})}^*\, (q_1,\bot,\ldots,\bot,\gamma v)
\end{equation}

\begin{equation}
\rho_2=(q_1,\bot,\ldots,\bot,\gamma v) \, \by{a}{}_{\m{T}(\m{M}_n)}\, (q_2,\gamma'\bot,\ldots,\bot, v)
\end{equation}

\begin{equation}
\rho_3=(q_2,\gamma'\bot,\ldots,\bot,v) \, \lby{\tau''}{}_{\m{T}(\m{M}_{[1,n[})}^*\, (q',\bot,\ldots,\bot,w')
\end{equation}

\medskip
\noindent 
for some $\tau', \tau'' \in \Sigma^*$ and $a \in \Sigma_{\epsilon}$ such that $\tau=\tau' a \tau''$.

\medskip

Observe that such decomposition of  $\rho$ is possible since in order to apply   a transition in $\Delta_n$,   the first $(n-1)$ stacks of $\m{M}$ must be  empty.

\medskip

\noindent
Now, we can apply the induction hypothesis to $\rho_1$, and we obtain the following run of   $\m{P}$:

\begin{equation}
\label{eq.2.32}
(q,w)\,  \lby{\epsilon}{}_{\m{T}(\m{P})}^* \,(q_1,\gamma v)
\end{equation}

\medskip

\noindent
We can also apply Lemma \ref{lemm-rest-n-1} to the run $\rho_3$, and we obtain that there is $u \in \Gamma^*$ such that:

\begin{equation}
u \in L(\m{M'}(q_2,\gamma',q')), \;\; \text{and} \;\; w'=u^R v
\end{equation}

\medskip
\noindent
From the run $\rho_2$, we get $\co q_1,\bot,\bot,\ldots,\bot,\gamma\cf\, \by{a}_{\m{M}}\, \co q_2,\gamma' \bot,\bot,\ldots,\bot,\epsilon\cf$. Moreover, we have $u \in L(\m{M'}(q_2,\gamma',q'))$. This implies that $u^R \in \delta(q_1,\gamma,\epsilon,q')$. This means that:

\begin{equation}
\label{eq.2.33}
(q,\gamma v) \by{\epsilon}{}_{\m{T}(\m{P})} (q',u^R v)=(q',w')
\end{equation}

\medskip
\noindent
Putting together Equations \ref{eq.2.32} and \ref{eq.2.33}, we obtain:

\begin{equation}
(q,w)\,  \lby{\epsilon}_{\m{T}(\m{P})}^*\,  (q,\gamma v) \by{\epsilon}{}_{\m{T}(\m{P})} (q',w')
\end{equation}

\medskip
\noindent
This terminates the If direction of Lemma  \ref{lemm.2.gpa->ompa}.
 \end{proof}

\medskip

Next, we prove that $\m{P}$ can perform a transition  from the initial configuration $(p_0,\bot) $ to a configuration of the form  $(q,w)$ if and only if  $\m{M}_{[1,n[}$ can move from the initial configuration $(q_0,\gamma_0 \bot,\bot,\ldots,\bot)$ to the  configuration  $(q,\bot,\bot,\ldots,\bot,w)$. 

\medskip

\begin{lem}
\label{lemm.2.intial.ompa=gpa}
For every $q \in Q$ and $w \in \mathit{Stack}(\m{M})$, $(p_0,\bot) \by{\epsilon}_{\m{T}(\m{P})}  (q,w)$  if and only if $(q_0,\gamma_0 \bot,\bot,\ldots,\bot)$ $ \lby{\tau}{}_{\m{T}(\m{M}_{[1,n[})}^*$ $(q,\bot,\bot,\ldots,\bot,w)$ for some $\tau \in \Sigma^*$.
\end{lem}

\medskip

\begin{proof} 
({\bf The If direction}) Assume that $\rho=(q_0,\gamma_0 \bot,\bot,\ldots,\bot)$ $ \lby{\tau}{}_{\m{T}(\m{M}_{[1,n[})}^*$ $(q,\bot,\bot,\ldots,\bot,w)$ for some $\tau \in \Sigma^*$. Then, we can apply Lemma  \ref{lemm-rest-n-1} to the run $\rho$. This implies that there is $u \in \Gamma^*$ such that:

\begin{center}
$(q_0,\gamma_0,\bot,\ldots,\bot) \lby{u}{}_{\m{T}(\m{M'})}^* (q,\bot,\ldots,\bot)$ \;\;and\;\;  $w=u^R \bot$
\end{center}

\medskip
\noindent
This means that $u \in L(\m{M'}(q_0,\gamma_0,q))$, and therefore  $u^R \bot \in \delta(p_0,\bot,\epsilon,q)$. This implies that the system $\m{T}({\m{P}})$ can move from  the configuration $(p_0,\bot)$ to the configuration $(q,u^R \bot)$:

$$(p_0,\bot) \by{\epsilon}{}_{\m{T}(\m{P})} (q,u^R\bot)=(q,w) $$

\medskip
\noindent
This terminates the proof of the If direction.

\medskip
\noindent
({\bf The Only if direction}) Assume that $(p_0,\bot) \by{\epsilon}_{\m{T}(\m{P})} (q,w)$. Then, from the definition of $\m{P}$, there is an $u \in \Gamma^*$ such that:

\begin{center}
$w=u^R \bot $\;\; and \;\;$u \in L(\m{M}'(q_0,\gamma_0,q))$
\end{center}

\medskip

From the   definition of $L(\m{M}'(q_0,\gamma_0,q))$, we have  $\rho=(q_0,\gamma_0 \bot,\bot,\ldots,\bot)$ $ \lby{u}{}_{\m{T}(\m{M'})}^* $ $(q,\bot,\bot,\ldots,\bot)$. We can apply Lemma \ref{lemm-rest-n-1} to the run $\rho$, and we obtain that there is $\tau \in \Sigma^*$ such that:

$$(q_0,\gamma_0\bot,\bot,\ldots,\bot) \lby{\tau}{}_{\m{T}(\m{M}_{[1,n[})}^* (q,\bot,\ldots,\bot,u^R \bot)= (q,\bot,\ldots,\bot,w) $$

\noindent
This terminates the Only if direction and the proof of Lemma \ref{lemm.2.intial.ompa=gpa}.
 \end{proof}

\medskip

Now, we are ready to prove  that the emptiness problem for $\m{M}$ is reducible to the emptiness problem for $\m{P}$.

\medskip

\noindent
 ({\bf The If direction}) Assume that $L(\m{P}) \neq \emptyset$. This implies that:

\begin{equation}
(p_0,\bot) \;\lby{\epsilon}{}_{\m{T}(\m{P})}^* \; (p_f,\epsilon) 
\end{equation}

\noindent
This means that there is a state $q \in F$ such that:

\begin{equation}
(p_0,\bot) \;\lby{\epsilon}{}_{\m{T}(\m{P})}^* \; (q,\bot) \by{\epsilon}{}_{\m{T}(\m{P})}  (p_f,\epsilon)
\end{equation}

\noindent
From the definition of the transition function of $\m{P}$,  there are  $q' \in Q$  and $w \in \mathit{Stack}(\m{M})$ such that:

\begin{equation}
\rho_1=(p_0,\bot)  \by{\epsilon}{}_{\m{T}(\m{P})} (q',w)
\end{equation}

\begin{equation}
\rho_2=(q',w)  \;\lby{\epsilon}{}_{\m{T}(\m{P})}^* \; (q,\bot) 
\end{equation}

\noindent
We can apply Lemma \ref{lemm.2.intial.ompa=gpa} to the run $\rho_1$, and we obtain that there is $\tau' \in \Sigma^*$ such that:

\begin{equation}
\label{2.41}
(q_0,\gamma_0 \bot,\ldots,\bot)\, \lby{\tau'}{}_{\m{T}(\m{M})}^*\, (q',\bot,\ldots,\bot,w)
\end{equation}

\noindent
We can also apply Lemma \ref{lemm.2.gpa->ompa} to the run $\rho_2$, and we obtain that there is $\tau'' \in \Sigma^*$ such that:

\begin{equation}
\label{2.42}
(q', \bot,\ldots,\bot,w)\, \lby{\tau''}{}_{\m{T}(\m{M})}^*\, (q,\bot,\ldots,\bot)
\end{equation}

\noindent
Putting together Equations \ref{2.41} and \ref{2.42}, we get:

\begin{equation}
(q_0,\gamma_0 \bot,\ldots,\bot)\, \lby{\tau'}{}_{\m{T}(\m{M})}^*\, (q', \bot,\ldots,\bot,w)\, \lby{\tau''}{}_{\m{T}(\m{M})}^*\, (q,\bot,\ldots,\bot)
\end{equation}

\noindent
This shows  that $L(\m{M}) \neq \emptyset$ since $q \in F$, and this terminates the proof of the If direction.

\medskip

\noindent
({\bf The Only if direction}:) Assume that $L(\m{M})\neq \emptyset$. Then,  there is a state $q \in F$ such that $(q_0,\gamma_0\bot,\bot,\ldots,\bot) $ $ \lby{\tau}{}_{\m{T}(\m{M})}^* $ $(q,\bot,\ldots,\bot)$ for some $\tau \in \Sigma^*$.

\noindent
This implies that there is a state $q' \in Q$,  $\tau', \tau'' \in \Sigma^*$, and $w \in \mathit{Stack}(\m{M})$ such that:

\begin{equation}
\rho_1=(q_0,\gamma_0 \bot,\ldots,\bot)\, \lby{\tau'}{}_{\m{T}(\m{M}_{[1,n[})}^*\, (q', \bot,\ldots,\bot,w)
\end{equation}

\begin{equation}
 \rho_2=(q', \bot,\ldots,\bot,w)\, \lby{\tau''}{}_{\m{T}(\m{M})}^*\, (q,\bot,\ldots,\bot)
\end{equation}

\noindent
We can apply Lemma \ref{lemm.2.intial.ompa=gpa} to the run $\rho_1$, and we obtain:

\begin{equation}
\label{2.43}
(p_0,\bot)\, \by{\epsilon}{}_{\m{T}(\m{P})}\, (q',w)
\end{equation}

\noindent
We can also apply Lemma \ref{lemm.2.gpa->ompa} to the run $\rho_2$, and we obtain that:

\begin{equation}
\label{2.44}
(q',w)\, \lby{\epsilon}{}_{\m{T}(\m{P})}^*\, (q,\bot)
\end{equation}

\noindent
Putting together Equations \ref{2.43} and \ref{2.44}, we get that:

\begin{equation}
(p_0,\bot)\, \by{\epsilon}{}_{\m{T}(\m{P})}\,(q',w)\, \lby{\epsilon}{}_{\m{T}(\m{P})}^*\, (q,\bot)
\end{equation}

\noindent
Moreover, we can apply the transition function $\delta(q,\bot,\epsilon,p_f)=\epsilon$ to the configuration $(q,\bot)$, and we obtain the following computation of $\m{T}(\m{P})$:

\begin{equation}
(p_0,\bot)\, \by{\epsilon}{}_{\m{T}(\m{P})}\,(q',w)\, \lby{\epsilon}{}^*_{\m{T}(\m{P})}\, (q,\bot)\, \by{\epsilon}{}_{\m{T}(\m{P})}\, (p_f,\epsilon)
\end{equation}

\noindent
This shows  that $L(\m{P}) \neq \emptyset$, and this terminates the proof of the Only if direction. \end{proof}


\subsection{Emptiness of a $n$-OMPA is in 2ETIME}
In the following, we show that the emptiness problem  for  a $n$-OMPA  is in 2ETIME with respect to the number of stacks. The proof is done by induction on the number of stacks. First, we use the induction hypothesis, that the emptiness problem for OMPA with $(n-1)$-stacks is decidable, to show that the generalized pushdown automaton $\m{P}$ is effective  (and so  the emptiness problem for $\m{P}$ is decidable). Once the effectiveness property of $\m{P}$ has been established, we  estimate the running time of  our saturation procedure for  $\m{P}$, given in section \ref{sat},  under the assumption that the emptiness problem.   for  $(n-1)$-OMPA can be solved in  2ETIME. We show that  the emptiness problem of $\m{P}$ (and so  $\m{M}$) is in 2ETIME

\begin{thm}
\label{ompa-complexity}
The emptiness problem  for an $n$-OMPA $\m{M}$ can be solved  in time $O(  {|\m{M}|}^{2^{dn}})$ for some constant $d$.
\end{thm}

\begin{proof}
Let $\m{M}=(n,Q,\Sigma,\Gamma,\Delta,q_0,\gamma_0,F)$ be an $n$-OMPA.    To prove Theorem \ref{ompa-complexity}, we proceed by induction on the number of stacks $n$.

\medskip
\noindent
{\bf Basis.} $n=1$. Then, $\m{M}$  is   a pushdown automaton. From \cite{BEM97}, we know  that the emptiness problem for $\m{M}$ can be solved in polynomial time  in $|\m{M}|$.

\medskip
\noindent
{\bf Step.} $n> 1$. Then, we can apply  Theorem \ref{ompa-gpa} to construct a  generalized pushdown automaton  $\m{P}=(P,\emptyset,\Gamma,\delta,p_0,\bot,\{p_f\})$, with $P=Q \cup \{p_0,p_f\}$, such that:

\begin{iteMize}{$\bullet$}

\item $L(\m{M})\neq \emptyset$ if and only if $L(\m{P}) \neq \emptyset$, and

\smallskip

\item For every $p_1,p_2 \in P$ and $\gamma \in \Gamma$,  there is an $(n-1)$-OMPA $\m{M}_{(p_1,\gamma,p_2)}$ over $\Gamma$  such that $L(\m{M}_{(p_1,\gamma,p_2)})=\big(\delta(p_1,\gamma,\epsilon,p_2)\big)^R$  and  $|\m{M}_{(p_1,\gamma,p_2)}|=O( |\m{M}|^2)$.

\end{iteMize}\medskip

\noindent It is easy to observe   that   $\m{P}$ is an effective generalized pushdown automaton. This is established by the following lemma.

\begin{lem}
\label{eff-gpa-ompa}
$\m{P}$ is an effective generalized pushdown automaton.

\end{lem}

\medskip

\begin{proof}
To prove the effectiveness property of $\m{P}$, we need to show that for every finite state automaton $\m{A}$ over the alphabet $\Gamma$, the problem of checking whether $L(\m{A}) \cap \delta(p_1,\gamma,\epsilon,p_2) \neq \emptyset$ is decidable for all $p_1,p_2 \in P$ and  $\gamma \in \Gamma$. It can be easy shown that $L(\m{A}) \cap \delta(p_1,\gamma,\epsilon,p_2) \neq \emptyset$ if and only if $L(\m{A})^R \cap (\delta(p_1,\gamma,\epsilon,p_2))^R \neq \emptyset$.

Let  $\m{A}$ be  a given  finite state automaton,   $p_1,p_2 \in P$ two states of $\m{P}$, and $\gamma \in \Gamma$ a stack symbol of $\m{P}$. Using   Lemma \ref{closurepropOMPA2}, we can  construct an $(n-1)$-OMPA $\m{M}'$  such that $L(\m{M}')=(L(\m{A}))^R \cap L(M_{(p_1,\gamma,p_2)})=(L(\m{A}))^R \cap (\delta(p_1,\gamma,\epsilon,p_2))$ since we have $M_{(p_1,\gamma,p_2)}=(\delta(p_1,\gamma,\epsilon,p_2))^R$. Now, we can   apply the induction hypothesis to $\m{M}'$ to show that the checking whether $L(\m{M}') \neq \emptyset$ is decidable. Thus,   $\m{P}$ is an effective generalized pushdown automaton.  \end{proof}

From Theorem  \ref{emptiness-EPDA}, Theorem \ref{ompa-gpa}, and  Lemma \ref{eff-gpa-ompa}, we deduce that the emptiness problem for the $n$-OMPA $\m{M}$ is decidable.  

Next, we will  estimate the running time of the decision procedure. From Theorem  \ref{emptiness-EPDA}, we know that  the emptiness problem of $\m{P}$  is reducible to compute  the  set of predecessors of the configuration $(p_f,\epsilon)$ since  $L(\m{P}) \neq \emptyset$ if and only if $(p_0,\bot) \in \mathit{Pre}_{\m{T}(\m{P})}^*(\{p_f\} \times \{\epsilon\})$.

Let $\m{A}$ be the $\m{P}$-automaton that recognizes the configuration $(p_f,\epsilon)$ of $\m{P}$. It  is easy to see that such $\m{P}$-automaton $\m{A}$, with $|\m{A}|=O(|\m{M}|)$, is effectively constructible. Now, we need to analysis the running time of the saturation procedure (given in section \ref{sat}) applied to $\m{A}$.  For that, let $\m{A}_0,\ldots, \m{A}_i$ be the sequence of $\m{P}$-automaton obtained from the saturation procedure such that $\m{A}_0=\m{A}$ and $L_{\m{P}}(\m{A}_{i})=\mathit{Pre}_{\m{T}(\m{P})}^*(L_{\m{P}}(\m{A}))$. Then, we have  $i =O ( |\m{M}|^3)$ since the number of possible new  transitions of $\m{A}$ is finite. Moreover, at each step $j$, with $0\leq j \leq i$, we need to check, for every state $q $ of $\m{A}$, $p,p' \in P$, and  $\gamma \in \Gamma$, whether $L(\m{A}_{j})(\{p'\},\{q\}) \cap \delta(p,\gamma,\epsilon,p') \neq \emptyset$.

Using   Lemma \ref{closurepropOMPA2}, we can    construct, in polynomial  time  in $|\m{M}|$, an $(n-1)$-OMPA $\m{M}'_{(q,p,\gamma,p')}$  such that $L(\m{M}'_{(q,p,\gamma,p')})=(L(\m{A}_j)(\{p'\},\{q\}))^R \cap L(M_{(p,\gamma,p')})$ and $|\m{M}'_{(q,p,\gamma,p')}| \leq c (|\m{M}|^3)$ for some constant $c$. Now, we can apply the induction hypothesis to $\m{M}'_{(q,p,\gamma,a,p')}$, and we obtain that the problem of  checking whether $L(\m{M}'_{(q,p,\gamma,a,p')})\neq \emptyset$ can be solved in time $O\big((c\,  |\m{M}|^3)^{2^{d(n-1)}}\big)$.  Putting together all these equations, we obtain that the problem of  checking whether $(p_0,\bot) \in \mathit{Pre}_{\m{T}(\m{P})}^*(\{p_f\} \times \{\epsilon\})$ can be solved in time $O\big(|\m{M}|^3 |\m{M}|^5 (c \, |\m{M}|^3)^{2^{d(n-1)}}\big)$. By  taking a constant $d$ as big  as needed, we can show that the problem of    checking whether $L(\m{M}) \neq \emptyset$ can be solved in time $O(|\m{M}|^{2^{d n}})$. \end{proof}



\section{Computing  the set of predecessors for  OMPA}

In this section, we show that the set of predecessors of a recognizable set $C$ of configurations of an OMPA is recognizable and effectively constructible (see Theorem \ref{coro.prestar-regularity}). To simplify the presentation, we can assume without loss of generality that  the set $C$ contains only one  configuration of the form $(q_f,\bot,\ldots,\bot)$ where all the stacks are empty. This result  is established by Lemma \ref{lemm.reduce-sec4}.

\begin{lem}
\label{lemm.reduce-sec4}
Let $\m{M}=(n,Q,\Sigma,\Gamma,\Delta,q_0,\gamma_0, F)$ be an OMPA and $\m{A}$ be an $\m{M}$-automaton. Then, it is possible  to  construct, in time and space polynomial in $(|\m{M}|+  |\m{A}|)$, an OMPA $\m{M}'=(n,Q'\cup \{q_f\},\Sigma,\Gamma',\Delta',q_0,\gamma_0, F)$  where $Q \subseteq Q'$, $q_f \notin Q'$,  and $|\m{M}'|=O(|\m{M}|  \cdot |\m{A}|)$ such that for every $c \in \mathit{Conf}(\m{M})$, $c \in Pre_{\m{T}(\m{M})}^*(L_{\m{M}}(\m{A}))$ if and only if  $c \in Pre_{\m{T}(\m{M}')}^*(\{(q_f,\bot,\ldots,\bot)\})$.\end{lem}

\begin{proof}
The proof  is similar to the case of standard pushdown automata. Technically, this can be done by adding to the OMPA $\m{M}$ some  pop transitions that check, in nondeterministic way, if the current configuration belongs to $L_{\m{M}}(\m{A})$ by  simulating the $\m{M}$-automaton $\m{A}$.  Let $\m{A}=(Q_{\m{M}}, \Gamma, \Delta_{\m{M}},I_{\m{M}},F_{\m{M}})$  be the $\m{M}$-automaton.  We assume w.l.o.g that $\m{A}$ has no transition leading to its initial states and that there is no transition of $\m{A}$ labeled by the empty word.

We construct the OMPA $\m{M}'=(n,Q'\cup \{q_f\},\Sigma,\Gamma',\Delta',q_0,\gamma_0, F)$ as follows:

\medskip

\begin{iteMize}{$\bullet$}

\item $Q'=Q \cup (Q_{\m{M}} \times [1,n+1])$. The set of states  $Q'$  is precisely the union of the   set of states of  $\m{M}$ and the set of states of $\m{A}$ indexed by the stack identities. (The index $n+1$ is used to mark the end of  the simulation  of $\m{A}$ by $\m{M}'$). Moreover, we assume that $\m{M}'$ has  a  fresh state $q_f \notin Q'$.

\medskip

\item $\Gamma'=\Gamma \cup \{\sharp\}$ such that $\sharp \notin \Gamma$. The fresh stack symbol $\sharp$ is used  to ensure that $\m{M}'$ respects the constraints imposed by the  normal formal. Intuitively, this symbol will be pushed on the first stack whenever a symbol is popped from a stack with an index from $2$ to $n$, during the simulation of the $\m{M}$-automaton $\m{A}$, and then this symbol  will be popped from the  first stack.

\medskip

\item  $\Delta'$ is the smallest transition relation  such that the following conditions are satisfied:

\medskip

\begin{iteMize}{$-$}

\item {\bf First Phase:}  In this phase the OMPA $\m{M}'$  behaves exactly as the OMPA $\m{M}'$.  This corresponds to    $\Delta \subseteq \Delta'$.

\medskip

\item {\bf Second phase:} In the second phase, $\m{M}'$ checks if the current configuration is accepted by $\m{A}$. This is done  by allowing $\m{M}'$ to start, in non-deterministically way,  the simulation of $\m{A}$ while popping  the read  symbols from their corresponding stacks. Formally, we have:

\medskip

\begin{iteMize}{$*$}

\item For every  transition $q\, \by{\gamma}_{\m{A}} \,p$ with $q \in Q$ and $\gamma \in (\Gamma \setminus \{\bot\})$, we have  $\co q,\gamma,\epsilon,\ldots,\epsilon \cf \by{\epsilon}_{\m{M}'}  \co (p,1), \epsilon,\epsilon,\ldots, \epsilon \cf$. This means that, in non-deterministically way, the checking  of whether the current configuration  is accepted   by $\m{A}$ can be  started by the  simulation  of a transition of $\m{A}$ from the current state $q$.  This transition of $\m{M}'$ corresponds to the case where the  first stack is not empty

\medskip

\item For every $q  \,\lby{\bot^{i} }{}_{\m{T}(\m{A})}^*  \,p$ with $i \in [1,n]$  and $q \in Q$, we have  $\co q,\bot,\epsilon,\ldots,\epsilon \cf \by{\epsilon}_{\m{M}'}  \co (p,i+1), \bot,\epsilon,\ldots, \epsilon \cf$.  This means that the simulation of $\m{A}$ by $\m{M}$ can be started and the first $i$ stacks are empty.  Observe that  the state $(p,i+1)$ of $\m{M}'$  corresponds to the fact that the current state of $\m{A}$ is $p$ and  that we are currently checking  the stack $i+1$.

\medskip

\item For every $i \in [1,n]$ and  $p  \,\lby{\bot^{j} }{}_{\m{T}(\m{A})}^*  \,p'$ for some $j \in [1,n-i+1]$, we have  $\co (p,i),\bot,\epsilon,\ldots,\epsilon \cf \by{\epsilon}_{\m{M}'}  \co (p',i+j), \bot,\epsilon,\ldots, \epsilon \cf$.  This corresponds to the simulation of a sequence of transitions of $\m{A}$ that checks if the stacks from $i$  to $(i+j-1)$ are empty. In  this case, we move the current state from $p$ to $p'$ and we start checking the stack  of index $(i+j)$.

\medskip

\item For every $i \in [1,n]$ and  $p  \,\by{\gamma }_{\m{A}}  \,p'$ for some $\gamma \in (\Gamma \setminus\{\bot\})$, we have  $\co (p,i),\gamma_1,\ldots,\gamma_n \cf \by{\epsilon}_{\m{M}'}  \co (p',i), \alpha_1,\ldots, \alpha_n \cf$ with $\gamma_1=\cdots=\gamma_{i-1}=\bot$, $\gamma_i=\gamma$, $\gamma_{i+1}=\cdots=\gamma_n=\epsilon$, $\alpha_1=\sharp \cdot \bot$, $\alpha_2=\cdots=\alpha_{i-1}=\bot$,  and $\alpha_{i}=\cdots=\alpha_n=\epsilon$.  The simulation of a   transition  of $\m{A}$, that reads the  symbol $\gamma$ from the $i^{th}$-stack,  is performed by $\m{M}'$ by a  transition that pops  $\gamma$ form the $i^{th}$-stack  and pushes the fresh symbol $\sharp$ into the first stack.

\medskip

\item For every index $i \in [1,n+1]$ and state $p \in Q_{\m{A}}$, we have  $\co (p,i),\sharp,\epsilon,\ldots, \epsilon \cf \by{\epsilon}_{\m{M}'}  \co (p,i), \epsilon,\ldots, \epsilon \cf$. This transition  pops the fresh symbol $\sharp $ from the first stack. Recall that this fresh  symbol is introduced in the only aim of ensuring the normal form of $\m{M}'$.

\medskip

\item For every   $p \in F$, we have  $\co (p,n+1),\bot,\epsilon,\ldots, \epsilon \cf \by{\epsilon}_{\m{M}'}  \co q_f, \bot,\epsilon,\ldots, \epsilon \cf$. This transition ends the simulation of $\m{A}$ by $\m{M'}$ after verifying that the current configuration of $\m{M}$ is accepted by $\m{A}$.

\end{iteMize}

\end{iteMize}
\end{iteMize}

\medskip

\noindent Then it is easy to see that for every $c \in \mathit{Conf}(\m{M})$, $c \in Pre_{\m{T}(\m{M})}^*(L_{\m{M}}(\m{A}))$ if and only if  $c \in Pre_{\m{T}(\m{M}')}^*(\{(q_f,\bot,\ldots,\bot)\})$.
\end{proof}

\medskip

In the following, we show  that the set  of configurations $C'$ of the form $(q',\bot,\ldots,\bot,w')$  from which the OMPA $\m{M}$ can  reach a configuration  of the form   $(q,\bot,\ldots,\bot)$, where all the stacks are empty, is recognizable and effectively constructible. 

\begin{lem}
\label{lemm-atig-ref}
Let $\m{M}=(n,Q,\Sigma,\Gamma,\Delta,q_0,\gamma_0, F)$ be an OMPA and  $q \in Q$ be a state. Then, it is possible to construct, in time   $O(|\m{M}|^{2^{d n}})$ where $d$ is a constant,  
an $\m{M}$-automaton $\m{A}$ such that $|\m{A}|=O(|\m{M}|)$ and $c \in L_{\m{M}}(\m{A})$ if and only if $c \in Pre_{\m{T}(\m{M})}^*(\{(q,\bot,\ldots,\bot)\})$ and $c=(q',\bot,\ldots,\bot,w)$ for some  $q' \in Q$ and $w \in \mathit{Stack}{(\m{M})}$.
\end{lem}

\begin{proof}
Lemma \ref{lemm.2.gpa->ompa}  shows  that, given an OMPA $\m{M}$ with $n$ stacks, it is possible  to construct an effective generalized pushdown automaton  $\m{P}$, whose  pushed languages are defined by   OMPA with $(n-1)$ stacks,  such that    the following invariant is preserved:  The state  and the stack's content of $\m{P}$ are the same as the state and the content of the $n^{th}$ stack of $\m{M}$ when its first $(n-1)$ stacks are empty.   Then, we can make use of   Theorem \ref{pred-EPDA}, which shows the set of all predecessors of a recognizable set of configurations is an effectively constructible recognizable set for effective generalized pushdown automata,    to show that Lemma \ref{lemm-atig-ref} holds. 
\end{proof}

\medskip

Next, we state our main theorem which is a  generalization of the result obtained in  for bounded-phase visibly multi-stack pushdown automata \cite{anil10}.

\begin{thm}
\label{prestar-regularity}
Let $\m{M}=(n,Q,\Sigma,\Gamma,\Delta,q_0,\gamma_0, F)$ be an OMPA and  $q \in Q$ be a state. Then, it is possible to construct, in time   $O(|\m{M}|^{2^{d n}})$ where $d$ is a constant,  
a $\m{M}$-automaton $\m{A}$ such that $|\m{A}|=O(|\m{M}|^{2^{d n}})$ and $L_{\m{M}}(\m{A})= \mathit{Pre}_{\m{T}(\m{M})}^*(\{(q,\bot,\ldots,\bot)\})$.
\end{thm}

\begin{proof} 
From Lemma \ref{lemm-atig-ref}, we know that   the set of configurations  $C_n$,  consisting of   $\mathit{Pre}_{\m{T}(\m{M})}^*(\{(q,\bot,\ldots,\bot)\})$  restricted to the configurations in which   the first $(n-1)$ empty stacks are empty, is   recognizable and effectively constructible.  To compute the intermediary configurations in  $\mathit{Pre}_{\m{T}(\m{M})}^*(\{(q,\bot,\ldots,\bot)\})$ where  the first $(n-1)$ stacks   are not empty, we construct an ordered multi-pushdown automaton  $\m{M}'$ with $(n-1)$ stacks that: $(1)$ performs  the same transitions on its stacks as the ones   performed by $\m{M}$ on its first $(n-1)$ stacks, and $(2)$ simulates a push transition  of $\m{M}$ over  its  $n^{th}$ stack  by a transition of  the finite-state automaton accepting  the recognizable  set of configurations $C_n$. Now, we can apply the induction hypothesis to $\m{M}'$ and construct a finite-state automaton accepting the set of all  predecessors  $\mathit{Pre}_{\m{T}(\m{M})}^*(\{(q,\bot,\ldots,\bot)\})$.

\medskip

 We proceed by induction on the number of stacks of the OMPA $\m{M}$.

\medskip
\noindent
{\bf Basis.} $n=1$. Then, $\m{M}$  is   a pushdown automaton. From \cite{BEM97}, we know  that such an $\m{M}$-automaton  $\m{A}$  for $\m{M}$ can be constructed in polynomial time  in $|\m{M}|$.

\medskip
\noindent
{\bf Step.} $n> 1$. Then, we can use Lemma \ref{lemm-atig-ref}   to construct, in time   $O(|\m{M}|^{2^{d n}})$ where $d$ is a constant, 
an $\m{M}$-automaton $\m{A}'=(Q_{\m{A}'},\Gamma,\Delta_{\m{A}'},Q,F_{\m{A}'})$ such that $|\m{A}'|=O(|\m{M}|)$ and $(q'',\bot,\ldots,\bot,w) \in L_{\m{M}}(\m{A}')$ if and only if $(q'',\bot,\ldots,\bot,w) \, \lby{\tau'}_{\mathcal{T}(\mathcal{M})}^*\, (q,\bot,\ldots,\bot) $
 for some $\tau' \in \Sigma^*$. Afterwards, we assume without loss of generality  that the  $\m{M}$-automaton has no $\epsilon$-transitions.

\medskip

Let  $\m{M}_{[1,n[}=(n,Q,\Sigma,\Gamma,\Delta_{[1,n[},q_{0},\gamma_{0},F)$ be the OMPA      built  from $\m{M}$ by discarding the set of pop transitions of $\m{M}$ over the $n^{th}$ stack (as defined in Section \ref{sect.ompan-1}).  (Recall that   $\Delta_{[1,n[}=\Delta \cap \big((Q \times  (\Gamma_{\epsilon})^{n-1} \times \{\epsilon\})\times \Sigma_{\epsilon}  \times (Q \times (\Gamma^*)^{n}) \big)$).   Then, it is easy to see that for every  configuration $(q',w_1,\ldots,w_n)$ in $Pre_{\m{T}(\m{M})}^*(\{(q,\bot,\ldots,\bot)\})$, there are $q'' \in Q$, $w \in \mathit{Stack}(\m{M})$, and $\tau', \tau\in  \Sigma^*$ such that:

$$(q',w_1,\ldots,w_n) \, \lby{\tau}{}_{\m{T}(\mathcal{M}_{[1,n[})}^* \, (q'', \bot, \ldots,\bot,w ) \,  \lby{\tau'}{}_{\m{T}(\m{M})}^*\,(q,\bot,\ldots,\bot)$$

\medskip

Since the OMPA $\m{M}_{[1,n[}$ can only  have push transitions over its $n^{th}$ stack,  we have $(q',w_1,\ldots,w_n)\, \lby{\tau}{}_{\m{T}(\m{M}_{[1,n[})}^*\, (q'', \bot,\ldots,\bot,w)$ if and only if  there is $v \in (\Gamma \setminus \{\bot\})^*$ such that $w=v w_n$ and $(q',w_1,\ldots,w_{n-1},\bot)\, \lby{\tau}{}_{\m{T}(\m{M}_{[1,n[})}^*\, (q'', \bot,\ldots,\bot,v)$ (see Lemma \ref{lemm-rest-n-1}).

\medskip

 Let  $\m{M}'=(n-1,Q \times Q_{\m{A}'},\Sigma,\Gamma,\Delta',q'_{0}, \gamma_{0}, F')$ be an $(n-1)$-OMPA built  from the OMPA $\m{M}_{[1,n[}$ and the FSA $\m{A}'$ such that  $\co (q_1,p_1),\gamma_1,\ldots,\gamma_{n-1} \cf \by{a}_{\m{M}'}
 \co (q_2,p_2),\alpha_1,\ldots,\alpha_{n-1} \cf$ if and only  if   
  $\co q_1,\gamma_1,\ldots,\gamma_{n-1},\epsilon \cf \by{a}_{\m{M}_{[1,n[}}
 \co q_2,\alpha_1,\ldots,\alpha_{n-1},\alpha_n \cf$ and  $p_2 \, \lby{\alpha_n}{}_{\m{T}(\m{A}')}^* \, p_1$ for some $ \alpha_{n} \in \big((\Gamma \setminus \{\bot\}) \cup \{\epsilon\}\big)$.
 In fact, the OMPA  $\m{M}'$ defines a kind of synchronous product between the pushed word over the $n$-th stack of  OMPA $\m{M}_{[1,n[}$ and the reverse of the input word of the FSA $\m{A}'$. Observe that the size of the constructed $(n-1)$- OMPA $\m{M}'$ is $O(|\m{M}|^2))$.

 \medskip
 
 Then, the relation between $\m{M}'$,  $\m{M}_{[1,n[}$, and $\m{A}'$ is given by  Lemma \ref{ref-lemma-sec6} which 
 follows immediately from the definition of $\m{M}'$.

 \begin{lem}
 \label{ref-lemma-sec6}
$ ((q_1,p_1),w_1,\ldots,w_{n-1})\, \lby{\varsigma}{}_{\m{T}(\m{M}')}^*\, ((q_2,p_2), \bot,\ldots,\bot)$  iff there is a $v \in (\Gamma \setminus \{\bot\})^*$ such that $(q_1,w_1,\ldots,w_{n-1},\bot)\, \lby{\tau}{}_{\m{T}(\m{M}_{[1,n[})}^*\, (q_2, \bot,\ldots,\bot,v \bot)$ and $p_2 \, \lby{v}{}_{\m{T}(\m{A}')}^* \, p_1$.
 \end{lem}

\medskip

 Now, we can apply the induction hypothesis  to  $\m{M}'$  to show that for every $(q'',p'') \in Q \times Q_{\m{A}'}$, it possible to construct, in time   $O(|\m{M}|^{2^{d (n-1)+2}})$,   
an $\m{M}'$-automaton $\m{A}_{(q'',p'')}$ such that $|\m{A}_{(q'',p'')}|=O(|\m{M}|^{2^{d (n-1)+2}})$ and $L_{\m{M}'}(\m{A}_{(q'',p'')})= Pre_{\m{T}(\m{M}')}^*(\{((q'',p''),\bot,\ldots,\bot)\})$.

\medskip

From the $\m{M}'$-automata $\m{A}_{(q'',p'')}$  and the $\m{M}$-automaton $\m{A}'$, we can  construct an $\m{M}$-automaton $\m{A}$ such that $(q',w_1,\ldots,w_n) \in L_{\m{M}}(\m{A})$ if and only if  there are $q'' \in Q$ and $p', p'' \in Q_{\m{A}}$ such that: $(1)$ $q'' \lby{\bot^{n-1}}{}_{\m{T}({\m{A}'})}^* p''$, $(2)$ $((q',p'),w_1,\ldots,w_{n-1}) \in L_{\m{M}'}(\m{A}_{(q'',p'')})$,  and  $(3)$
  $p' \lby{w_n}{}_{\m{T}({\m{A}'})}^* p$ for some $p \in F_{\m{A}'}$. Observe that such an automaton $\m{A}$ of the  size  $O(|\mathcal{M}|^{2^{d n}})$ (by taking $d$ as big as needed)  is effectively constructible from $\m{A}_{(q'',p'')}$  and $\m{A}'$ using standard automata operations.  Moreover, we have:
  
  \begin{lem}
  \label{lemm-sec6-proof}
  $L_{\m{M}}(\m{A})= Pre_{\m{T}(\m{M})}^*(\{(q,\bot,\ldots,\bot)\})$.
  \end{lem}

\begin{proof}
{\bf ($\subseteq$)}  Let $(q',w_1,\ldots,w_n) \in L_{\m{M}}(\m{A})$. Then, there are  $q'' \in Q$ and $p', p'' \in Q_{\m{A}}$ such that: $(1)$ $q'' \lby{\bot^{n-1}}{}_{\m{T}({\m{A}'})}^* p''$, $(2)$ $((q',p'),w_1,\ldots,w_{n-1}) \in L_{\m{M}'}(\m{A}_{(q'',p'')})$, and  $(3)$
  $p' \lby{w_n}{}_{\m{T}({\m{A}'})}^* p$ for some $p \in F_{\m{A}'}$.

  So, we can apply Lemma \ref{ref-lemma-sec6} to the run  $ ((q',p'),w_1,\ldots,w_{n-1})\, \lby{\varsigma}{}_{\m{T}(\m{M}')}^*\, ((q'',p''), \bot,\ldots,\bot)$ to show that there is $v\in \Gamma^*$ such that $(q',w_1,\ldots,w_{n-1},\bot)\, \lby{\tau}{}_{\m{T}(\m{M}_{[1,n[})}^*\, (q'', \bot,\ldots,\bot,v)$ and $p'' \, \lby{v}{}_{\m{T}(\m{A}')}^* \, p'$. 
  Thus, we have   $(q',w_1,\ldots,w_{n-1},w_n)\, \lby{\tau}{}_{\m{T}(\m{M})}^*\, (q'', \bot,\ldots,\bot,vw_n)$.

Now, we can use the runs $q'' \lby{\bot^{n-1}}{}_{\m{T}({\m{A}'})}^* p''$, $p'' \, \lby{v}{}_{\m{T}(\m{A}')}^* \, p'$, and  $p' \lby{w_n}{}_{\m{T}({\m{A}'})}^* p$  to show that $(q'',\bot,\ldots,\bot,vw_n) \in L_{\m{M}}(\m{A'})$. This implies that $(q'',\bot, \ldots,\bot, v w_n)\, \lby{\tau'}{}_{\m{T}(\m{M})}^*\,(q,\bot,\ldots,\bot)$.

Hence, we have   $(q',w_1,\ldots,w_{n})\in Pre_{\m{T}(\m{M})}^*(\{(q,\bot,\ldots,\bot)\})$ and therefore  $L_{\m{M}}(\m{A})\subseteq Pre_{\m{T}(\m{M})}^*(\{(q,\bot,\ldots,\bot)\})$.

\medskip
\noindent
{\bf ($\supseteq$)}
 Let $(q',w_1,\ldots,w_n) \in Pre_{\m{T}(\m{M})}^*(\{(q,\bot,\ldots,\bot)\})$. Then, there are $q'' \in Q$, $v \in \Gamma^*$, and $\tau, \tau' \in \Sigma^*$ such that:

$$(q',w_1,\ldots,w_n)\, \lby{\tau}{}_{\m{T}(\m{M}_{[1,n[})}^*\, (q'',\bot,\ldots,\bot, v w_n)\, \lby{\tau'}{}_{\m{T}(\m{M})}^*\,(q,\bot,\ldots,\bot) $$

Since  $(q'', \bot, \ldots, \bot, v w_n) \lby{\tau'}{}^*_{\mathcal{T}(\mathcal{M})} (q, \bot, \ldots, \bot)$,  we have $(q'', \bot , \ldots, \bot, v w_n) \in L_{\m{M}}(\m{A}')$. This implies that there are $p',p'' \in Q_{\m{A}'}$ and $p \in F_{\m{A}'}$ such that $q'' \lby{\bot^{n-1}}{}_{\m{T}({\m{A}'})}^* p''$, $p'' \, \lby{v}{}_{\m{T}(\m{A}')}^* \, p'$, and  $p' \lby{w_n}{}_{\m{T}({\m{A}'})}^* p$. 

On the other hand,  we can show  $(q',w_1,\ldots,w_{n-1},\bot)\, \lby{\tau}{}_{\m{T}(\m{M}_{[1,n[})}^*\, (q'', \bot,\ldots,\bot,v)$ since we have $(q',w_1,\ldots,w_n)\, \lby{\tau}{}_{\m{T}(\m{M}_{[1,n[})}^*\, (q'', \bot,\ldots,\bot,vw_n)$. 

Then, 
 we can  apply Lemma \ref{ref-lemma-sec6} to  $(q',w_1,\ldots,w_{n-1},\bot)\, \lby{\tau}{}_{\m{T}(\m{M}_{[1,n[})}^*\, (q'', \bot,\ldots,\bot,v)$ and $p'' \, \lby{v}{}_{\m{T}(\m{A}')}^* \, p'$ to show that   $ ((q',p'),w_1,\ldots,w_{n-1})\, \lby{\varsigma}{}_{\m{T}(\m{M}')}^*\, ((q'',p''), \bot,\ldots,\bot)$. This implies that $((q',p'),w_1,\ldots,w_{n-1}) \in L_{\m{M}'}(\m{A}_{(q'',p'')})$. Now, we can use the definition of the $\m{M}$-automaton $\m{A}$ to show that $(q',w_1,\ldots,w_n) \in L_{\m{M}}(\m{A})$ since we have $q'' \lby{\bot^{n-1}}{}_{\m{T}({\m{A}'})}^* p''$,  $((q',p'),w_1,\ldots,w_{n-1}) \in L_{\m{M}'}(\m{A}_{(q'',p'')})$, and 
  $p' \lby{w_n}{}_{\m{T}({\m{A}'})}^* p$ with  $p \in F_{\m{A}'}$. Hence, we have 
$L_{\m{M}}(\m{A})\supseteq Pre_{\m{T}(\m{M})}^*(\{(q,\bot,\ldots,\bot)\})$. \end{proof}


\medskip
\noindent
This terminates the proof of Theorem \ref{prestar-regularity}.   \end{proof}

As an immediate consequence of Theorem \ref{prestar-regularity} and Lemma \ref{lemm.reduce-sec4}, we obtain:

\begin{thm}
\label{coro.prestar-regularity}
Let $\m{M}=(n,Q,\Sigma,\Gamma,\Delta,q_{0},\gamma_{0},F)$ be an OMPA and  $\m{A}'$ be an  $\m{M}$-automaton. Then, it is possible to construct, in time   $O((|\m{M}| \cdot |\m{A}'|)^{2^{d n}})$ where $d$ is a constant,  
an $\m{M}$-automaton $\m{A}$ such that $|\m{A}|=O((|\m{M}| \cdot |\m{A}'|)^{2^{d n}})$ and $L_{\m{M}}(\m{A})= Pre_{\m{T}(\m{M})}^*(L_{\m{M}}(\m{A}'))$.

\end{thm}

We  can extend the previous result to show that the operator $\mathit{Pre}^+$   preserves also   recognizability. 

\begin{thm}
\label{coro.prestar+regularity}
Let $\m{M}=(n,Q,\Sigma,\Gamma,\Delta,q_{0},\gamma_{0},F)$ be an OMPA and  $\m{A}'$ be an  $\m{M}$-automaton. Then, it is possible to construct, in time   $O((|\m{M}|\cdot |\m{A}'|)^{2^{d n}})$ where $d$ is a constant,  
an $\m{M}$-automaton $\m{A}$ such that $|\m{A}|=O((|\m{M}|\cdot |\m{A}'|)^{2^{d n}})$ and $L_{\m{M}}(\m{A})= Pre_{\m{T}(\m{M})}^+(L_{\m{M}}(\m{A}'))$.

\end{thm}

\begin{proof}
 In the following, we show that computing   the set $\mathit{Pre}^+_{\m{T}(\m{M})}(L_{\m{M}}(\m{A}'))$ can be reduced to computing the set $\mathit{Pre}^*_{\m{T}(\m{M'})}(L_{\m{M'}}(\m{A''}))$ for an OMPA $\m{M}'$ and an $\m{M}$-automaton  $\m{A}''$ such that $|\m{M}'|=O(|\m{M}|)$ and $|\m{A}''|=O(|\m{A}|)$. Intuitively, the OMPA has the same  stack and input alphabets as the ones of $\m{M}$. Corresponding to each state $q$ of $\m{M}$,   $\m{M}'$ has  $q$ and $q_{\mathit{copy}}$ as two states where $q_{\mathit{copy}}$ is a fresh symbol which was not used neither in the definition of $\m{M}$ nor in the definition of $\m{A}$. For any transition of the form $\co q,\gamma_1,\ldots,\gamma_n\cf \by{a}{}_{\m{M}} \co q',\alpha_1,\ldots,\alpha_n\cf$, $\m{M'}$ has two transitions $\co q,\gamma_1,\ldots,\gamma_n\cf \by{a}{}_{\m{M}'} \co q',\alpha_1,\ldots,\alpha_n\cf$ and $\co q,\gamma_1,\ldots,\gamma_n\cf \by{a}{}_{\m{M}'} \co q'_{\mathit{copy}},\alpha_1,\ldots,\alpha_n\cf$.  Any computation of $\m{M}'$ can be divided in two phases. In the first phase $\m{M}'$ mimics the behavior of the  OMPA $\m{M}'$ by performing the same sequence of transitions. In the second phase, the OMPA $\m{M}'$ performs a  transition from a state $q \in Q$  of $\m{M}$   to a state $q'_{\mathit{copy}}$ with $q' \in Q$ and halts. Formally, $\m{M}'$ is defined by the tuple $(n,Q \cup Q_{\mathit{copy}}, \Sigma, \Gamma,\Delta \cup  \Delta',q_0,\gamma_0,F)$ where $Q_{\mathit{copy}}=\{q_{\mathit{copy}} \,|\, q \in Q\}$ and $\Delta'= \{ \co q,\gamma_1,\ldots,\gamma_n\cf \by{a}{}_{\m{M}'} \co q'_{\mathit{copy}},\alpha_1,\ldots,\alpha_n\cf \,|\,   \co q,\gamma_1,\ldots,\gamma_n\cf \by{a}{}_{\m{M}} \co q',\alpha_1,\ldots,\alpha_n\cf\}$.
 
 Let $\m{A}'=(Q_{\m{M}},\Gamma,\Delta_{\m{M}},I_{\m{M}},F_{\m{M}})$ be the $\m{M}$-automaton. We assume here that $\m{A}$ has no transition leading to an initial state. 
 Now, we can construct the $\m{M}'$-automaton $\m{A''}=(Q_{\m{M}'},\Gamma,\Delta_{\m{M}'},I_{\m{M}'},F_{\m{M}'})$ from  the $\m{M}$-automaton $\m{A}'$  as follows: The set of states of $\m{A}''$  is the union of the set of states of $\m{A}'$ and the set of states of $\m{M}'$ (i.e., $Q_{\m{M}'}=Q_{\m{M}} \cup Q_{\mathit{copy}}$). The set of transitions of $\m{A}''$ contains any transition of $\m{A}$ that  does not involve a state of $Q$ (i.e., $ (\Delta_{\m{M}} \setminus (Q \times \Gamma_{\epsilon} \times Q_{\m{M}})) \subseteq \Delta_{\m{M}'}$). Moreover, corresponding to any transition  of the form $q \by{a}_{\m{A}'} p$ where $q \in Q$, the automaton $\m{A}''$ has a transition of the form  $q_{\mathit{copy}} \by{a}_{\m{A}''} p$. That is, the automaton $\m{A}''$ is precisely $\m{A}'$ where any initial state $q \in Q$ is relabeled by its copy $q_{\mathit{copy}}$. (Observe that there is no transition from/to a state $q \in Q$  in $\m{A}''$.) Since $\m{A}''$ is an $\m{M}'$-automaton, we have $I_{\m{M}'}=Q \cup Q_{\mathit{copy}}$. We have also $F_{\m{M}'}=F_{\m{M}}$.

Then it is easy to see that the set  $\mathit{Pre}^+_{\m{T}(\m{M})}(L_{\m{M}}(\m{A}'))$ is precisely the set $\mathit{Pre}^*_{\m{T}(\m{M}')}(L_{\m{M}'}(\m{A}'')) \cap \mathit{Conf}(\m{M})$. Thus, we can apply Theorem  \ref{coro.prestar-regularity} to show  that it is   possible to construct, in time   $O((|\m{M}|\cdot |\m{A}'|)^{2^{d n}})$ where $d$ is a constant,   an $\m{M}$-automaton $\m{A}$ such that $|\m{A}|=O((|\m{M}|\cdot |\m{A}'|)^{2^{d n}})$ and $L_{\m{M}}(\m{A})= Pre_{\m{T}(\m{M})}^+(L_{\m{M}}(\m{A}'))$.
\end{proof}

\section{Linear-Time  Global Model Checking}
In this section, we show that the model-checking problem  of $\omega$-regular properties for OMPA is decidable and in 2ETIME. Observe that this result subsumes the 2ETIME upper  bound obtained for the emptiness problem of OMPA (see Theorem \ref{ompa-complexity}). In fact, we can see  the emptiness problem (i.e., the reachability problem) of OMPA as a particular instance of the LTL-model checking problem of OMPA for which the decision procedure (provided in Section \ref{chap2-dir1-2ETIME}) is simpler.

To prove the 2ETIME upper bound for the model-checking problem  of $\omega$-regular properties for OMPA,  we introduce the repeated state  reachability problem for OMPA.

We fix an OMPA   $\m{M}=(n,Q,\Sigma,\Gamma,\Delta,q_{0},\gamma_{0},F)$ for the rest of the paper such that $\Sigma=\Delta$ and $t=((q,\gamma_1,\ldots,\gamma_n),a,(q',\alpha_1,\ldots,\alpha_n))$  is in $\Delta$ if and only if $a=t$.

\subsection{The repeated state  reachability problem}

In the following, we are interested in  solving  {\em the repeated state  reachability problem} which consists in computing, for  a given   state  $q_f \in Q$, the  set of  all configurations  $c$ of  $\m{M}$ such that   there is an infinite run  of $\m{T}(\m{M})$ starting from $c$ that visits infinitely often the state $q_f$.

To this aim,  let us   introduce the following notation: For every $i \in [1,n]$, we denote by $\m{M}_{[1,i]}=(n,Q,\Sigma,\Gamma,\Delta_{[1,i]},q_{0},\gamma_{0},F)$  the OMPA built from $\m{M}$ by discarding   pop transitions  of $\m{M}$  over the  last $(n-i)$ stacks. Formally, we have $\Delta_{[1,i]}=\Delta \cap \big( \big(Q \times (\Gamma_{\epsilon})^{i} \times (\{\epsilon\})^{n-i} \big) \times \Sigma \times \big(Q \times (\Gamma^*)^n\big) \big)$.

 For every  $i \in [1,n]$, and every $(q,\gamma) \in Q \times (\Gamma \setminus \{\bot\})$, let $C_i^{(q,\gamma)}$ denote the set of all configurations $(q,w_1,\ldots,w_n) \in \mathit{Conf}({\m{M}})$ such that  $w_1= \cdots =w_{i-1}=\bot$ and $w_i= \gamma u$ for some $u \in  \mathit{Stack}(\m{M})$. Moreover,  let $c_i^{(q,\gamma)}$ be the configuration  $(q,w_1,\ldots,w_n) $  of $\m{M}$ such that   $w_i= \gamma \bot$  and $w_j=\bot$ for all $j \neq i$.
 
In the following, we show that detecting  an infinite computation of $\m{M}$  that visits infinitely often a state $q_f$   can be reduced to detecting an infinite computation  the form $\rho_1\cdot  \rho_2^{\omega}$  that eventually repeats the same sequence of transitions indefinitely and visits $q_f$ infinitely often  (and  where $\rho_1$ and $\rho_2$ are finite computations). Hence, a periodic computation is a run which, after a finite computation prefix $\rho_1$ (called stem),  ultimately repeats the same sequence of transitions  $ \rho_2$ (called lasso) over and over. Let us give some intuitions behind this reduction.

 Let us assume that  there is an infinite computation $\rho$ of $\m{T}(\m{M})$ starting from a configuration $c$ of $\m{M}$. Let $i$ be the maximal  index of the stack that is popped infinitely often. This means that, at some point of the computation, the stacks from $(n-i+1)$ to $n$ will never be popped. Let us concentrate on  the suffix of the computation $\rho$ which contains only push transitions on  the stacks  from $(n-i+1)$ to $n$.  Let    $c_1 c_2 \cdots $ be the sequence of configurations in this suffix of $\rho$ where  the first $(i-1)$ stacks  are empty. Applying a similar argument to the content  of the $i^{th}$-stack, along the sequence of configurations  $c_1 c_2 \cdots $,  as the one for standard pushdown automata \cite{BEM97}, we can deduce that the $i^{th}$-stack is {\em increasing}. This means that there are  indices  $j_1 <j_2$, a stack symbol $\gamma \in \Gamma$, and a state $q \in Q$ such that the configurations  $c_{j_1}$ and $c_{j_2}$ are  in $C_i^{(q,\gamma)}$ and the  symbol $\gamma$ at the top of the $i^{th}$-stack in $c_{j_1}$ and $c_{j_2}$ will never be popped. Furthermore, along the sub-computation $\rho_2$ from $c_{j_1}$ and $c_{j_2}$ the state $q_{f}$ is visited.  Observe that if we remove from the configuration $c_{j_1}$ all the stack symbols that will never be popped we obtain the configuration $c_i^{(q,\gamma)}$. Then, 
 the computation $\rho_2$ can be simplified, by dropping all the (useless) stack symbols that will never be popped from the configuration $c_{j_1}$, as follows : $c_i^{(q,\gamma)}  \lby{\tau_1}{}^+_{\m{T}(\m{M}_{[1,i]})}  (q_f, w_1, \ldots, w_2)  \lby{\tau_2}{}^*_{\m{T}(\m{M}_{[1,i]})} c'_{2}$ with $c'_2 \in C_i^{(q,\gamma)}$. This computation $\rho_2$ represents our lasso computation since from the configuration $c'_2$ we can repeat the sequence of transitions $\tau_1 \tau_2$ while visiting the same states (and in particular the state $q_f$).

 The existence of a such lasso is expressed by  the second item of Theorem \ref{sect-ltl-lemma} while the existence of a stem computation from the starting configuration $c$ to the configuration $c_{j_1} \in C_i^{(q,\gamma)}$ is stated by  the first  item of Theorem \ref{sect-ltl-lemma}.

  Then, the solution of the   {\em repeated state  reachability problem} is formally based on the following fact:

 \begin{thm}
 \label{sect-ltl-lemma}
 Let $c$ be a configuration of $\m{M}$  and $q_f$ be a state of $\m{M}$. There is an infinite run starting from $c$ that visits infinitely often the state $q_f$ if and only if there are  $i \in [1,n]$, $q \in Q$, and $\gamma \in \Gamma$ such that:

 \begin{enumerate}
 \item  $c \in \mathit{Pre}_{\m{T}(\m{M})}^*(C_i^{(q,\gamma)})$, and 
 \item $c_i^{(q,\gamma)} \in \mathit{Pre}_{\m{T}(\m{M}_{[1,i]})}^+\big( \mathit{Pre}_{\m{T}(\m{M}_{[1,i]})}^*(C_i^{(q,\gamma)}) \cap (\{q_f\} \times (\mathit{Stack}(\m{M}))^n)\big)$.
\end{enumerate}
 \end{thm}

\begin{proof}

$(\Ra):$ Let $\rho=c_0 t_0 c_1 t_1 c_2 t_2 \cdots$ be an infinite computation  of $\m{T}(\m{M})$ starting from the configuration $c_0=c$ of $\m{M}$. For every $j \in \mathbb{N}$,  $c_j$ is a configuration of $\m{M}$ and $t_j$ is a transition of $\m{M}$ such that $ c_j \,\by{t_{j}}{}_{\m{T}(\m{M})} \, c_{j+1}$ (recall that $\Sigma=\Delta$).  Let $i \in [1,n]$ be the maximal index 
 such that for every $j \in \mathbb{N}$, there is $k_j \geq j$ such that  $t_{k_j}$ is a pop transition over the $i^{th}$ stack of $\m{M}$. This implies that  $c_{k_j}$  is in $Q \times (\{\bot\})^{i-1} \times ((\Gamma \setminus \{\bot\})^* \cdot \mathit{Stack}(\m{M})) \times (\mathit{Stack}(\m{M}))^{n-i}$ (i.e., the first $(i-1)$-stacks are empty) since $t_{k_j}$ is a pop transition  from the $i^{th}$ stack of $\m{M}$.

From the  definition of $i$, there is $r \in \mathbb{N}$ such that for every $h\geq r$, there is $d_h \in [1,i]$ such that  the transition $t_h$ is a pop transition over the  stack $d_h$ of $\m{M}$ (i.e., the transition $t_h$ is not  a pop transition  from  the stack from $(n-i+1)$ to $n$).  This implies that for every $h\geq r$, we have $c_{h} \by{t_h}{}_{\m{T}(\m{M}_{[1,i]})} \, c_{h+1}$.

Then,  we construct a sequence  $\pi= c_{j_0} c_{j_1} c_{j_2} \cdots$ of configurations of $\m{M}$   as follows: $c_{j_0}$ is the first configuration of $\rho$ such that $j_0 \geq r$ and $t_{j_0}$ is a pop transition over the $i^{th}$-stack of $\m{M}$, for every $\ell >0$, $c_{j_{\ell}}$ is the first configuration of $\rho$ such that $j_{\ell}> j_{\ell-1}$ and $t_{j_{\ell}}$ is a pop transition over the $i$-stack of $\m{M}$. Recall that, by definition, we have for every $l \in \mathbb{N}$, $c_{j_{l}}$  is in $Q \times (\{\bot\})^{i-1} \times ((\Gamma \setminus \{\bot\})^* \cdot \mathit{Stack}(\m{M})) \times (\mathit{Stack}(\m{M}))^{n-i}$ (i.e.,  the first $(i-1)$ stacks are empty).

 Now, for every $l \geq 0$, let $\pi^{(l)}$ be the suffix of $\pi$ starting at $c_{j_l}$, and let $m^{(l)}$ be the minimal length of the configurations of $\pi^{(l)}$, where the length of a configuration is defined as the length of its $i^{th}$ stack.

Construct a subsequence $\pi'=c_{z_0} c_{z_1} c_{z_2} \cdots  $ of $\pi$ as follows: $c_{z_0}$ is the first configuration of $\pi$ of length $m^{(0)}$; for every $l >0$, $c_{z_l} $ is the first configuration of $\pi^{(z_{l-1}+1)}$ of length $m^{(z_{l-1}+1)}$.

 Since the number of states and stack symbols is finite, there exists a subsequence $\pi''=c_{x_0} c_{x_1} c_{x_2} \cdots  $ of $\pi'$ whose elements have all the same state $q$, and the same symbol $\gamma$ on the top of the $i^{th}$ stack. Observe that $c_{x_0}, c_{x_1}, c_{x_2}, \ldots$ are in $C_i^{(q,\gamma)}$.

 Since $\rho$ is an accepting run, there is an index $b\geq 1$ and a configuration $c_{q_f}$ with state $q_f$ such that:
 $$ c_0 \lby{\tau}{}_{\mathcal{T}(\mathcal{M})}^* \, c_{x_0} \, \lby{\tau'}{}_{\mathcal{T}(\mathcal{M})}^+ \, c_{q_f} \lby{\tau''}{}_{\m{T}(\m{M})}^*\, c_{x_b}$$

 Since $c_0=c$ and $c_{x_0} \in C_i^{(q,\gamma)}$, we have $c \in \mathit{Pre}_{\m{T}(\m{M})}^*(C_i^{(q,\gamma)})$, and  so $(1)$ holds.

 Due to the definition of $\pi$ (and so, $\pi'$ and $\pi''$), we have 
 
 $$c_{x_0} \lby{\tau'}{}_{\mathcal{T}(\mathcal{M}_{[1.i]})}^+  c_{q_f} \lby{\tau''}{}_{\m{T}(\m{M}_{[1,i]})}^*\, c_{x_b}$$

Since $c_{x_0} \in Q \times (\{\bot\})^{i-1} \times ((\Gamma \setminus \{\bot\})^* \cdot \mathit{Stack}(\m{M})) \times (\mathit{Stack}(\m{M}))^{n-i}$, then  there are $w_i,w_{i+1},\ldots,w_n \in \mathit{Stack}(\m{M})$ such that  $c_{x_0}=(q,\bot,\ldots,\bot,\gamma w_i, w_{i+1},\ldots,w_n)$.  Due to the definition of the subsequence $\pi'$ and $\pi''$ all the configurations of $\rho$ between $c_{x_0}$ and $c_{x_b}$ have a content of the $l$-th stack (with $i\leq l \leq k$) of the form $w'_l w_l$. In particular, the configuration $c_{q_f}$ is of the form $(q_f,u_1,\ldots,u_{i-1}, u_i w_i, u_{i+1} w_{i+1},\ldots, u_n w_{n})$ and the configuration $c_{x_b}$ is of the form $(q,\bot,\ldots,\bot, \gamma v_i w_i, v_{i+1} w_{i+1},\ldots, v_n w_{n})$. This implies:

$$c_i^{(q,\gamma)}=(q,\bot,\ldots,\bot,\gamma,\bot,\ldots,\bot)  \lby{\tau'}{}_{\m{T}(\m{M}_{[1,i]})}^+  (q_f,u_1,\ldots,u_{i-1}, u_i, u_{i+1} ,\ldots, u_n)  $$

 and 
$$ (q_f,u_1,\ldots,u_{i-1}, u_i, u_{i+1} ,\ldots, u_n) \lby{\tau''}{}_{\m{T}(\m{M}_{[1,i]})}^*\,(q,\bot,\ldots,\bot, \gamma v_i, v_{i+1},\ldots, v_n )$$

Consequently, $(2)$ holds, which concludes the proof.

\medskip

\noindent
$(\Leftarrow):$ We can use $(1)$ and $(2)$ of Theorem \ref{sect-ltl-lemma} to construct a run starting from $c$ that visits infinitely often the state $q_f$.\end{proof}

 Since the sets of configurations $C_i^{(q,\gamma)}$ and  $(\{q_f\} \times (\mathit{Stack}(\m{M}))^n)$ are   recognizable, we can use Theorem \ref{coro.prestar-regularity} and Theorem \ref{coro.prestar+regularity} to construct $\m{M}$-automata recognizing   $\mathit{Pre}_{\m{T}(\m{M})}^*(C_i^{(q,\gamma)})$ and $ \mathit{Pre}_{\m{T}(\m{M}_{[1,i]})}^+\big( \mathit{Pre}_{\m{T}(\m{M}_{[1,i]})}^*(C_i^{(q,\gamma)}) \cap (\{q\} \times (\mathit{Stack}(\m{M}))^n)\big)$. Hence, we can construct a $\m{M}$-automaton that recognizes the set of all configurations $c$ of  $\m{M}$ such that   there is an infinite run  of $\m{T}(\m{M})$ starting from $c$ that visits infinitely often the state $q_f$.   

\medskip

\begin{thm}
\label{repeated-thm}
Let $\m{M}=(n,Q,\Sigma,\Gamma,\Delta,q_{0},\gamma_{0},F)$ be an OMPA and  $q_f \in Q$ be a  state. Then, it is possible to construct, in time   $O((|\m{M}|)^{2^{d n}})$ where $d$ is a constant,  
an $\m{M}$-automaton $\m{A}$ such that $|\m{A}|=O((|\m{M}|)^{2^{d n}})$  and  for every configuration $c \in \mathit{Conf}(\m{M})$,  $c \in L_{\m{M}}(\m{A})$   if and only if   there is an infinite run  of $\m{T}(\m{M})$ starting from $c$ that visits $q_f$ infinitely often.  
\end{thm}

\smallskip

\begin{proof}
We know from  Theorem \ref{sect-ltl-lemma}  that there is an infinite run  of $\m{T}(\m{M})$ starting from $c$ that visits $q_f$ infinitely often if and only if  $c \in \mathit{Pre}_{\m{T}(\m{M})}^*(C_i^{(q,\gamma)})$, and 
$c_i^{(q,\gamma)} \in \mathit{Pre}_{\m{T}(\m{M}_{[1,i]})}^+\big( \mathit{Pre}_{\m{T}(\m{M}_{[1,i]})}^*(C_i^{(q,\gamma)}) \cap (\{q_f\} \times (\mathit{Stack}(\m{M}))^n)\big)$. (Observe that it is possible to construct  an $\m{M}$-automaton representing the set $C_i^{(q,\gamma)}$ and which size is linear in the size of $\m{M}$.)

Then, for every index $i \in [1,n]$, state $q \in Q$ and stack symbol $\gamma \in \Gamma$, we construct an $\m{M}$-automaton $\m{A}_1^{(q,i,\gamma)}$ recognizing the set $ \mathit{Pre}_{\m{T}(\m{M}_{[1,i]})}^*(C_i^{(q,\gamma)})$ and  such that $|\m{A}_1^{(q,i,\gamma)}|=O((|\m{M}|)^{2^{d' n}})$  where $d'$ is a constant. From Theorem \ref{coro.prestar-regularity}, we know that such an automaton $\m{A}_1^{(q,i , \gamma)}$  can be constructed  in time   $O((|\m{M}|)^{2^{d' n}})$.

Now, we can construct an  $\m{M}$-automaton $\m{A}_2^{(q,i,\gamma)}$ recognizing precisely  the set $\big( \mathit{Pre}_{\m{T}(\m{M}_{[1,i]})}^*(C_i^{(q,\gamma)}) \cap (\{q_f\} \times (\mathit{Stack}(\m{M}))^n)\big)$ and  such that $|\m{A}_2^{(q,i,\gamma)}|=O((|\m{M}|)^{2^{d' n}})$. Observe that $\m{A}_2^{(q,i,\gamma)}$ can be constructed   in time   $O((|\m{M}|)^{2^{d' n}})$ from the $\m{M}$-automaton   $\m{A}_1^{(q,i,\gamma)}$.

We can apply  Theorem \ref{coro.prestar+regularity} to $\m{M}$ and  $\m{A}_2^{(q,i,\gamma)}$ to show that we can construct an $\m{M}$-automaton $\m{A}_3^{(q,i,\gamma)}$ recognizing $ \mathit{Pre}_{\m{T}(\m{M}_{[1,i]})}^+\big( \mathit{Pre}_{\m{T}(\m{M}_{[1,i]})}^*(C_i^{(q,\gamma)}) \cap (\{q_f\} \times (\mathit{Stack}(\m{M}))^n)\big)$ and such that $|\m{A}_3^{(q,i,\gamma)}|=O((|\m{M}|)^{2^{d'' n}})$ for some constant $d''>d'$. Moreover, such an $\m{M}$-automaton $\m{A}_3^{(q,i,\gamma)}$ can be constructed in time   $O((|\m{M}|)^{2^{d'' n}})$. Then, checking whether $c_i^{(q,\gamma)}$  is in $L_{\m{M}}(\m{A}_3^{(q,i,\gamma)})$ can be performed in time polynomial in  $|\m{A}_3^{(q,i,\gamma)}|$  \cite{hu79}.

If $c_i^{(q,\gamma)}$ is not  in  $L_{\m{M}}(\m{A}_3^{(q,i,\gamma)})$ then let $\m{A}^{(q,i,\gamma)}$ be the $\m{M}$-automaton recognizing the empty set (i.e., $L_{\m{M}}(\m{A}^{(q,i,\gamma)})=\emptyset$). Otherwise let $\m{A}^{(q,i,\gamma)}$ be the $\m{M}$-automaton recognizing the set $\mathit{Pre}_{\m{T}(\m{M})}^*(C_i^{(q,\gamma)})$ and such that $|\m{A}^{(q,i,\gamma)}|=O((|\m{M}|)^{2^{d' n}})$. From Theorem \ref{coro.prestar-regularity}, we know that such an automaton $\m{A}^{(q,i , \gamma)}$  can be constructed  in time   $O((|\m{M}|)^{2^{d' n}})$.

By taking $d$ as big as needed, we can construct, in time   $O((|\m{M}|)^{2^{d n}})$ where $d$ is a constant,   the $\m{M}$-automaton $\m{A}$ such that $|\m{A}|=O((|\m{M}|)^{2^{d n}})$  and  for every configuration $c \in \mathit{Conf}(\m{M})$,  $c \in L_{\m{M}}(\m{A})$   if and only if   there is an infinite run  of $\m{T}(\m{M})$ starting from $c$ that visits $q_f$ infinitely often.  The $\m{M}$-automaton $\m{M}$ is just the union of all the $\m{M}$-automata $\m{A}^{(q,i,\gamma)}$.
\end{proof}

 \medskip

 \subsection{$w$-regular  properties}
 
In the following, we assume that the reader is familiar with $w$-regular properties expressed in the  linear-time temporal logics \cite{Pnu77} or the linear time $\mu$-calculus \cite{Var88}. For more details, the reader is
referred to \cite{Pnu77,VW86,Var88, Var95}.

Let   $\varphi$  be  an $w$-regular  formula built from  a set of atomic propositions $\mathit{Prop}$,  and  let  $\m{M}=(n,Q,\Sigma,\Gamma,\Delta,q_{0},\gamma_{0},F)$ be an  OMPA  with a labeling function $\Lambda: \, Q \rightarrow 2^{\mathit{Prop}}$  associating  to each state $q \in Q$  the set of atomic propositions that are true in  it.  Afterwards, we are interested in solving   {\em the global model checking problem}  which   consists in   computing the set of all configurations $c$ of $\m{M}$  such that every infinite run starting from $c$ satisfies    $\varphi$.

To solve this problem, we adopt an  approach similar to \cite{BM96,BEM97} and we construct   a Buchi   automaton $\mathcal{B}_{\neg \varphi}$ over the alphabet  $2^{\mathit{Prop}}$ accepting the negation of $\varphi$ \cite{VW86,Var95}. Then, we compute the product of the OMPA $\m{M}$ and of the B$\ddot{u}$chi automaton $\m{B}_{\neg \varphi}$ to obtain an $n$-OMPA $\m{M}_{\neg \varphi}$ with a set of repeating states $G$. Now, it is easy to see  that the original problem can be reduced  to the {\em repeated state  reachability problem}  which  compute the set of all configurations $c$ such that there  is an infinite run  of $\m{T}(\m{M})$ starting from $c$ that visits infinitely often a state in $G$.   Hence,  as an immediate consequence of Theorem \ref{repeated-thm}, we obtain:

\begin{thm}
\label{thm-ltlt-}
Let $\m{M}=(n,Q,\Sigma,\Gamma,\Delta,q_{0},\gamma_{0},F)$ be an OMPA with a labeling function $\Lambda$,  and  let  $\varphi $  be  a linear time $\mu$-calculus formula or linear time temporal formula. Then, it is possible to construct, in time   $O(( 2^{|\varphi|} \cdot |\m{M}|)^{2^{d n}})$ where $d$ is a constant,  
an $\m{M}$-automaton $\m{A}$ such that $|\m{A}|=O((2^{|\varphi|} \cdot |\m{M}|)^{2^{d n}})$  and  for every configuration $c \in \mathit{Conf}(\m{M})$,  $c \in L_{\m{M}}(\m{A})$   if and only if   there is an infinite run  of $\m{T}(\m{M})$ starting from $c$  does not satisfy $\varphi$.  \end{thm}

\begin{proof}
It is well known that it is possible to construct, in time exponential  in $|\varphi|$,   a B$\ddot{u}$chi automaton $\m{B}_{\neg \varphi}$ for the negation of ${ \varphi}$ having exponential size in $|\varphi|$ \cite{VW86,Var88}. Therefore,  the product of $\m{M}$ and $\m{B}_{\neg \varphi}$ has polynomial size in $|\m{M}|$ and exponential size in $|\varphi|$. Applying  Theorem \ref{repeated-thm}  to the $n$-OMPA $\m{M}_{\neg \varphi}$ (the product  of $\m{M}$ and $\m{B}_{\neg \varphi}$) of size $O(2^{|\varphi|} \cdot |\m{M}|)$  we obtain our complexity result. \end{proof}

Observe that we can also  construct an  $\m{M}$-automaton $\m{A}'$  such that for every configuration $c \in \mathit{Conf}(\m{M})$,  $c \in L_{\m{M}}(\m{A})$   if and only if   every  infinite run  of $\m{T}(\m{M})$ starting from $c$ that satisfies $\varphi$ since the class of $\m{M}$-automata is closed under boolean operations.

We are now ready to establish our result about the model checking problem for  $w$-regular properties  which consists in   checking whether, for a given  configuration $c$ of  the OMPS,  every infinite run starting from $c$ satisfies   the formula $\varphi$.
\begin{thm}
The model checking problem for the linear-time temporal logics or the linear-time $\mu$-calculus and OMPA is 2ETIME-complete.
\end{thm}
\begin{proof}
The 2ETIME upper bound is established  by Theorem \ref{thm-ltlt-}. To prove  hardness, we use the fact that the emptiness problem for ordered multi-pushdown automata is 2ETIME-complete \cite{ABH-dlt08}.   
\end{proof}


\section{Conclusion}
We  have shown that  the set of all  predecessors of a recognizable set of configurations  of an ordered multi-pushdown automaton  is an effectively constructible recognizable set. We  have also  proved that the set of all configurations of an ordered multi-pushdown automaton   that satisfy a given $w$-regular property is effectively recognizable. From these results  we have derived  an 2ETIME upper bound for the model checking problem of $w$-regular properties.

It may be interesting to see if our approach can be extended to solve the global  model-checking problem  for branching time properties expressed in CTL or CTL$^*$ by adapting the constructions given in  \cite{BEM97,AF-BW-PW-INF-97} for standard pushdown automata.

\bibliographystyle{alpha}

\bibliography{biblio}

\end{document}